\documentclass[preprint,12pt,authoryear]{elsarticle}

\usepackage{amsfonts}
\usepackage{amsmath}
\usepackage{bm}
\usepackage{mathtools}
\usepackage{xspace}
\usepackage{makecell}
\usepackage[reftex]{theoremref}
\usepackage[ruled,lined,linesnumbered,longend]{algorithm2e}
\usepackage{url}
\usepackage{hyperref}

\newtheorem{theorem}{Theorem}[section]
\newtheorem{lemma}[theorem]{Lemma}
\newtheorem{corollary}[theorem]{Corollary}

\newtheorem{proposition}[theorem]{Proposition}
\newtheorem{definition}[theorem]{Definition}

\newtheorem{example}[theorem]{Example}
\newtheorem{proof}[theorem]{Proof}

\newenvironment{example1}
  {\begin{example}\rm}
  {\end{example}}

\journal{}

\newcommand{\wgm}{Shang et. al's general method\xspace}

\newcommand{\ignore}[1]{}

\newcommand{\overbar}[1]{\mkern 1.5mu\overline{\mkern-1.5mu #1 \mkern-1.5mu}\mkern 1.5mu}

\newcommand{\rX}{\mathbb{R}[X]}

\newcommand{\rXY}{\mathbb{R}[\overbar{X}]}
\newcommand{\qX}{\mathbb{Q}[X]}

\newcommand{\aX}{\mathbb{A}[X]}
\newcommand{\aXX}{\mathbb{A}[X_1, \dots, X_n]}
\newcommand{\aXY}{\mathbb{A}[\overbar{X}]}

\DeclarePairedDelimiter{\ceil}{\lceil}{\rceil}

\DeclarePairedDelimiter{\set}{\lbrace}{\rbrace}

\DeclareMathOperator{\quadmod}{QM}
\DeclareMathOperator{\preord}{PO}
\DeclareMathOperator{\semialgebraic}{\mathcal{S}}

\DeclareMathOperator{\ord}{ord}

\DeclareMathOperator{\compact}{compact}

\DeclareMathOperator{\algFindEps}{\sf FindEps}
\DeclareMathOperator{\algFindK}{\sf FindK}
\DeclareMathOperator{\algBasicLemma}{\sf BasicLemma}
\DeclareMathOperator{\algSOSBasicLemma}{\sf SOSBasicLemma}

\DeclareMathOperator{\algRemoveStrictPosBet}{\sf RemoveStrictPosBetween}
\DeclareMathOperator{\algRemoveStrictPosLeft}{\sf RemoveStrictPosLeft}

\DeclareMathOperator{\algPutinar}{\sf Certificate}

\DeclareMathOperator{\algSaturatedCert}{\sf SaturatedCert}

\DeclareMathOperator*{\argmax}{arg\,max}

\DeclareMathOperator{\Nat}{Nat}
\DeclareMathOperator{\PO}{PO}
\DeclareMathOperator{\QM}{QM}
\DeclareMathOperator{\Pos}{Pos}

\newcommand{\maple}{\texttt{Maple}\xspace}

\newcommand{\lpq}{\texttt{liftPO2QM}\xspace}
\newcommand{\rc}{\texttt{RealCertify}\xspace}
\newcommand{\csq}{\texttt{CertSatQM}\xspace}

\newtheorem{lemmaalt}{Lemma}

\newenvironment{lemmar}[1]{
  \renewcommand\thelemmaalt{#1}
  \lemmaalt
}{\endlemmaalt}

\SetKwInput{kwRequires}{Requires}
\SetKwInput{kwEnsures}{Ensures}
\SetKwInput{kwInput}{Input}
\SetKwInput{kwOutput}{Output}
\SetKw{kwGoto}{Go to}
\SetKw{kwBreak}{break}
\SetKw{kwContinue}{continue}



\begin{document}

\begin{frontmatter}



\title{Computing Certificates in Archimedean Univariate Saturated Quadratic Modules} 


\author[1]{Jose Abel Castellanos Joo\corref{cor1}}
\ead{jabelcastellanosjoo@unm.edu}

\author[1]{Deepak Kapur}
\ead{kapur@cs.unm.edu}

\cortext[cor1]{Corresponding author}

\affiliation[1]{organization={Computer Science Department, University of New Mexico},
    addressline={1155 University Blvd SE},
    city={Albuquerque},
    country={United States of America}}

\date{}

\begin{abstract}

  A new symbolic algorithm to compute sums of squares multipliers (certificates)
  to witness
  the membership of non-negative univariate polynomials in a saturated
  univariate quadratic module is presented. Certificates are first computed in
  terms of natural generators introduced by Kuhlmann and Marshall for an
  Archimedean saturated quadratic module; natural generators can be easily
  read-off from a semialgebraic set. In the univariate case, an Archimedean
  quadratic module is also a preordering since it is closed under multiplication;
  certificates have different representations when a polynomial is viewed as a
  member in a quadratic module versus in a preordering
  An algorithm is given to compute certificates of natural generators in terms
  of the original generators; it uses a construction introduced by Kuhlmann,
  Marshall, and Schwartz known as the ``Basic Lemma'', which splits the
  non-negative factors of generators.
  To compute a quadratic module certificate, certificates of products of natural
  generators are computed using a detailed case analysis based on the types of
  natural generators.

  An implementation of the algorithms proposed in \maple is also
  discussed. The certificates obtained using this implementation are compared
  with those generated by \rc. We discuss examples where \rc is unable to find
  certificates while the proposed method is successful.
\end{abstract}

\begin{keyword}
Saturated quadratic modules \sep Natural generators \sep Algebraic Certificates


\end{keyword}

\end{frontmatter}



\section{Introduction}

Ever since Hilbert posed the 17th problem of Positivstellens\"atz, which was
settled by Artin \cite{Artin}, there has been considerable interest in finding
algorithmic implementable solutions for special cases. Noteworthy are the
results by Schm\"udgen for preorderings \cite{Schmdgen1991TheKmomentPF}, Putinar
for quadratic modules \cite{10.2307/24897130}, developing representations of
positive and non-negative polynomials on a semialgebraic set defined by a
finite set of generators. The complexity results in
\cite{10.1016/j.jco.2006.07.002,SCHWEIGHOFER2002307} gave upper bounds on the
degree of multipliers that appear in the representations of a polynomial in
terms of the original generators. However, there has been limited progress in
developing implementable algorithmic solutions for computing certificates
exhibiting membership of a polynomial in a quadratic module; a noteworthy
exception is \rc, a package in \maple. This paper addresses this problem for the
special case of generating membership certificates of a univariate polynomial in
an Archimedean saturated quadratic module generated by finite generators of
univariate polynomials. In particular, an algorithm is given to compute
certificates, which has been implemented as a \maple package and compared to
\rc.

The importance of certificates that exhibit the details of a result produced by
a computer algebra system has long been emphasized, especially by Kahan,
Fateman, \cite{Fateman2000ImprovingEI} and recently Davenport
\cite{10.1007/978-3-031-42753-4_17}, among others. Certificate computation has
received particular attention in the SMT (Satisfiability Modulo Theories)
\cite{BdMF15, 10.1007/978-3-031-10769-6_3}, theorem proving
\cite{gregoire2011proof}, and more recently, in deep learning networks
\cite{katz2017reluplex, singh2018fast, singh2019abstract} due to the ability of
such software to solve very large problems in numerous applications in
engineering, software, hybrid systems, and other safety-critical systems. The
dependence on reliable results produced by such software in life-critical
applications further emphasizes the need for certificates that show the
correctness of the results.

Given a finite set $G =\set{g_1, \dots, g_s}$ of polynomials in a polynomial
ring $\aXX$ (also abbreviated as $\aXY$), where
$\mathbb{A} \subseteq \mathbb{R}$ is the set of real algebraic numbers. Let
$\semialgebraic(G)$ denote the (basic) semialgebraic set defined by $G$, i.e.,
$\semialgebraic(G) := \set*{\bm{x} \in \mathbb{R}^n \mid g(\bm{x}) \geq 0, g \in
  G}$; given $S \subseteq \mathbb{R}$, let $\Pos(S)$ (resp. $\Pos^{+}(S)$)
denote the set of non-negative polynomials (resp. strictly positive) over $S$,
i.e., $\Pos(S) := \set{f \in \rXY \mid f(\bm{x}) \geq 0, \forall \bm{x} \in S}$;
finally, let $\sum{\aXY}^2$ denote
$\set*{\sum_{i=1}^{n}{f_i^{2}} \mid f_i \in \aXY}$.

The Schm\"udgen's Positivstellens\"atz theorem \cite{Schmdgen1991TheKmomentPF}
states that for any strictly positive polynomial $f$ over a compact
$\semialgebraic(G)$, $f \in \preord(G)$, where $\preord(G)$ is the set of all
sums \\
$\set*{\sum_{}^{}{\sigma_e \cdot g_1^{e_1} \cdots g_s^{e_s} \mid \sigma_e \in
    \sum{\rXY}^2}}$, $e = (e_1, \dots, e_s) \in \set{0, 1}^{s}$ and each
$\sigma_e$ is a sum of squares.  Putinar further simplified Schm\"udgen's
representation showing that for any strictly positive polynomial $f$ over a
compact $\semialgebraic(G)$, if $\quadmod(G)$ is Archimedean then
$f \in \quadmod(G)$, where the quadratic module $\quadmod(G)$ generated by $G$
is the set of all sums
$\set*{\sigma_0 + \sum_{i=1}^{s}{\sigma_i \cdot g_i \mid \sigma_i \in
    \sum{\rXY}^2}}$.  For a quadratic module to be \textbf{Archimedean}, it must
contain a polynomial of the form $N - \sum_{i=1}^{n}{X_i^{2}}$ where
$N \in \mathbb{N}^{+}$. The sums of squares multipliers $\sigma_e$'s in the
above sets are called the \textbf{certificates}\footnote{We abuse the
  terminology somewhat and interchange to call a tuple of $\sigma_e$'s,
  certificates or certificate in the paper.}  that witness the membership of
polynomials in these structures.

By definition, preorderings are closed under multiplication, whereas quadratic
modules are generally not. Several necessary conditions for quadratic modules to
satisfy a preordering structure were studied in \cite{ScheidererDistinguished2005}.
Marshall and his collaborators \cite{marshall2002,marshall2005} studied several
variations of Schm\"udgen's theorem. In particular, the authors found the
necessary conditions for these results to hold in quadratic modules, as well as
established equivalences between these conditions.

In case an Archimedean quadratic module is a preordering also, which is the case of
univariate structures, the certificate problem for a given polynomial in such a
structure has a subtle distinction, depending upon whether the certificate is
being sought using its representation in a quadratic module or a preordering.

A algorithm for computing certificates of a polynomial in the univariate case
for an Archimedean saturated quadratic module is given. The algorithm uses {\bf
  natural (choice)} generators introduced by Kuhlmann and Marshall
\cite{marshall2002,marshall2005, marshallbooksos} as a way to describe
generators of a \textbf{saturated} quadratic module to include all 
non-negative polynomials in $\rX$ over a given semialgebraic set
$\semialgebraic(G)$. We emphasize that a key novelty of the approach handles
non-negative polynomials in contrast to other approaches \cite{MAGRON2021221,
  Baldi2023DegreeBF, Baldi_2022, SCHWEIGHOFER2002307, 10.1007/s10957-012-0261-9}
that cover the case for strictly positive polynomials derived by Putinar's
Positivstellens\"atz \cite{10.2307/24897130}. Other lines of research cover
similar cases and consider different assumptions to our problem in
consideration; in \cite{10.1145/3476446.3535480}, the authors find weighted sums
of Hermitian squares decomposition of nonnegative trigonometric polynomials with
Gaussian coefficients over the unit circle using complex root isolation and
semidefinite programming techniques.  In \cite{baldi2024effective}, the authors
address the problem of computing certificates of nonnegative polynomials over
finite semialgebraic sets in the context of multivariate polynomials; however, it 
requires additional assumptions. In particular, the ideal $I$ associated with
the equality constraints is zero dimensional and that the ideals $(f)$ and
$(I : f)$ are coprime where $f$ is the input polynomial.

To compute a certificate in terms of the natural generators of
$\semialgebraic(G)$, a construction in Augustine's
thesis\cite{Augustin2008TheMP} is adapted, which gives a certificate using the
preordering representation; to get a certificate in the quadratic module
representation, certificates of products of natural generators must be
constructed; this is shown in Section \ref{cert_prods}.  Section
\ref{repr_nat_gens} describes a method for computing certificates of natural
generators in terms of the original set of generators.

The operations in the presented algorithms are effective, as it only manipulates
real algebraic numbers which are computable \cite{L70c,
  10.5555/1197095}. Finally, an algorithm is given to compute the certificates
of each element in the split generators in terms of the original
generators. Combining these results, an algorithm is given for computing a
certificate of a given polynomial in terms of the original generators.

The presented construction uses a lemma from \cite[Theorem 3.2, page
590]{marshall2005} to split factors of members and a theorem from
\cite[Proposition 4.8, page 1058]{scheiderer2000} to find sums of squares
multipliers that avoid products between generators; We describe algorithms to
perform constructively these results in Appendix
\ref{basiclemma-appendix}. Section \ref{experiments} discusses an implementation
in \maple and compares several examples with the \rc package, a \maple package
that computes certificates in Archimedean quadratic modules.

 \section{Preliminaries}

This paper deals with univariate structures, so all polynomials henceforth are
univariate. Assume that the semialgebraic set of $G := \set{g_1, \dots, g_s}$ is
\footnote{Since $G \subseteq \qX$, the end points of $\semialgebraic(G)$ are
  real algebraic numbers.}  $\semialgebraic(G) = \bigcup_{i=0}^{k}{[a_i, b_i]}$
with $a_j \leq b_j, b_{j-1} < a_j$ for $j = 1, \dots, k$. 

\subsection{Constructive results about quadratic modules}

In \cite{marshall2005}, the authors provide constructive proofs of the
following results. We use them in our constructions, so we adapt some
terminology to fit our context in quadratic modules in the univariate case:

\begin{theorem} \cite[Basic Lemma (Lemma 2.1)]{marshall2005}
  \th\label{sec:results-} Given $f, g \in \Pos(\semialgebraic(G))$ being
  relatively prime, there exist strictly positive polynomials $\sigma, \tau$
  over $\semialgebraic(G)$ such that $1 = \sigma f + \tau g$.
\end{theorem}

This theorem is proved constructively by explicitly computing both $\sigma$ and
$\tau$. As an observation, by Schm\"udgen's Positivstellens\"atz, both $\sigma$
and $\tau$ belong to the preordering $\PO(G)$. Under additional conditions, $\tau$
and $\sigma$ become sums of squares as the following theorem says:

\begin{theorem} \cite[Corollary 2.3 (i)]{marshall2005}
  \th\label{sec:results--1}
  Given  $f, g \in \Pos(\semialgebraic(f
  g))$ being relatively prime, where $\semialgebraic(f g)$ is compact,   
  there exists $\sigma, \tau \in \sum{\rX}^2$ such that
  $1 = \sigma f + \tau g$. 
\end{theorem}

This theorem is used to obtain certificates in quadratic modules of the factors
of members satisfying the above conditions. The fact that $\sigma, \tau$ can be
made to be sums of squares, allows us to avoid products in the preordering
structure. Similarly to this result, there is another constructive approach in
\cite[Proposition 4.8]{scheiderer2000} using Weierstrass-approximation that
computes $\sigma, \tau$ as sums of squares under equivalent conditions. In
 Appendix \ref{basiclemma-appendix} we discuss an algorithm $\algBasicLemma$
(resp. $\algSOSBasicLemma$) to compute $\sigma, \tau$ in
\th\ref{sec:results-} (resp. \th\ref{sec:results--1}.)

\begin{lemma}\cite[Lemma 1.5]{Augustin2008TheMP}
  \th\label{sec:cert-elem-ponats}
  If $a, c_1, c_2, b \in \mathbb{R}$ with $a \leq c_1 \leq c_2 \leq b$ then
  \begin{equation}
    \label{sec:cert-elem-ponats-1}
    \begin{split}
      (X - c_1)(X - c_2) \in \quadmod((X - a)(X - b))
    \end{split}
  \end{equation}
  Moreover, there exists a non-negative number $\gamma \in \mathbb{R}$ depending
  on $a, c_1, c_2, b$ such that
  $(X - c_1)(X - c_2) = \sigma_0 + \gamma \cdot (X-a)(X-b)$ where
  $\sigma_0 \in \sum{\rX}^2$.
\end{lemma}

It is worth noting the proof of \th\ref{sec:cert-elem-ponats} is constructive
and finds an arithmetical expression \footnote{The set of arithmetical
  expressions is a set closed by the arithmetic operators (addition,
  subtraction, multiplication, division, and exponentiation.)} for
$\gamma$. Hence, if the parameters $a, c_1, c_2, b$ are computable, $\gamma$
above is computable as well.

\subsection{Natural Generators associated with $\semialgebraic(G)$}
\label{sec:natur-gener-assoc-1999}

In general, the
saturation $\tilde{Q}$ of a quadratic module $Q$ is defined as the smallest
intersection of all quadratic modules, including $Q$. A quadratic module $Q$ is
saturated if $\tilde{Q} = Q$. The following theorem provides a semantic characterization:

\begin{theorem}\cite[][Proposition 2.6.1, page 33]{marshall2008positive}
  \th\label{sec:preliminaries-1}
  Let $G$ be a finite subset of $\rX$ and $Q := \quadmod(G)$. Then,
  \begin{enumerate}
  \item $\tilde{Q} = \{f \in \rX \mid f \geq 0 \text{ over }
    \semialgebraic(G)\}$. 
  \item $Q$ is saturated if and only if for all $f \in \rX$, $f \geq
    0$ over $\semialgebraic(G)$ implies $f \in Q$.
  \end{enumerate}
\end{theorem}

For the rest of our paper, the quadratic module $\QM(G)$ is assumed to
be saturated.

Kuhlmann and Marshall introduced the concept of natural generators associated
with a compact semialgebraic set $\semialgebraic(G)$; we adapt their
terminology.

\begin{definition}\cite{marshall2002}
  \th\label{sec:preliminaries} For a given semialgebraic set $S \neq \emptyset$,
  the natural generators of $S$, denoted $\Nat(S)$, is the set of polynomials
  including:
  \begin{itemize}
  \item[(i)] If $a \in S$ and $(-\infty, a) \cap S = \emptyset$, then
    $X - a \in \Nat(S)$. This is called the left natural generator or simply the
    {\bf left linear factor}.
  \item[(ii)] If $b, c \in S, b < c, (b, c) \cap S = \emptyset$, then
    $(X - b)(X - c) \in \Nat(S)$.  This is called a {\bf quadratic
      factor}.
  \item[(iii)] If $d \in S$ and $(d, \infty) \cap S = \emptyset$, then
    $-(X - d) \in \Nat(S)$. This is called the right natural generator, or simply
    the {\bf right linear factor}.
 
  \item[(iv)] $\Nat(S)$ has no other elements except these.
  \end{itemize}

  If $S = \emptyset$, then the set of natural generators is $\{-1\}$.
\end{definition}

This theorem relates the concept of saturated quadratic modules and natural
generators, we focus only on the compact case:

\begin{theorem}\cite[][Theorem 3.1, item (b)]{marshall2005}
  \th\label{sec:natur-gener-assoc}
  If $\semialgebraic(G)$ is compact, the following are equivalent:
  \begin{itemize}
  \item $\quadmod(G)$ is saturated.
  \item $\quadmod(G)$ contains the natural generators of $\semialgebraic(G)$.
  \end{itemize}
\end{theorem}

From the latter, we can see that if $\semialgebraic(G)$ is empty, the quadratic
module $\quadmod(G)$ is saturated. The following theorem from
\cite{marshall2005} gives a practical procedure to determine whether $QM(G)$ is
saturated.

\begin{theorem}\cite[Theorem 3.2]{marshall2005}
  \th\label{sec:cert-elem-ponats-2} $\QM(G)$ is saturated if and only if
  $\semialgebraic(G)$ satisfies the following two conditions:
  \begin{itemize}
  \item[(i)] for each left endpoint $a_{j}$ in $\semialgebraic(G)$, there exists
    $i \in \set{1, \dots, s}$ such that $g_i(a_j) = 0$ and
    $\frac{d g_i}{dX}(a_j) > 0$,
  \item[(ii)] for each right endpoint $b_j$ in $\semialgebraic(G)$, there exists
    $i \in \set{1, \dots, s}$ such that $g_i(b_j) = 0$ and
    $\frac{d g_i}{dX}(b_j) < 0$.
  \end{itemize}
\end{theorem}

A method to generate a certificate of $f$ in $\PO(\Nat(S))$ is extracted
from the inductive argument about the degree of $f$ in \cite{Augustin2008TheMP};
it involves computing certificates of factors of $f$ that are non-negative over
$S$. Since preorderings are closed by multiplication, the certificate of $f$ is
obtained by multiplying and rearranging the certificates of each of its factors
appropriately in the preordering structure.

\begin{theorem}\cite[Theorem 1.6]{Augustin2008TheMP}
  \th\label{sec:background} For a basic closed semialgebraic set
  $S \subseteq \mathbb{R}$ we have $\Pos(S) = \PO(\Nat(S))$. Furthermore, the
  sums of squares multipliers in the preordering representation are computable.
\end{theorem}

The certificates in \th\ref{sec:background} are obtained using an inductive
proof on the degree of the given polynomial. The base case is when the
polynomial is non-negative \footnote{This is because, in the univariate case,
  non-negative polynomials correspond to sums of squares polynomials which are
  members of any quadratic module.} over $\mathbb{R}$; the inductive case is
when the given polynomial evaluates to a negative value, say $f(c) < 0$ for some
$c \in \mathbb{R}$, and $\deg(f) > 0$. This is divided into three cases:

\begin{itemize}
\item Case 1: If $c \leq a_0$, then there exists $c_0$ such that
  $c < c_0 \leq a$ such that $f(c_0) = 0$. Hence, $f = (X-c_0)g$ for some $g$
  that is non-negative over $S$.
\item Case 2: If $c \geq b_k$, then there exists $c_0$ such that
  $b_k \leq c_0 < c$ such that $f(c_0) = 0$. Hence, $f = (-(X-c_0)) g$ for some
  $g$ is non-negative over $S$.
\item Case 3: If $b_i < c < a_{i+1}$, then there exists $c_0, c_1$ such that
  $b_i \leq c_0 \leq c_1 \leq a_{i+1}$ and $f(c_0) = f(c_1) = 0$. Hence,
  $f = (X-c_0)(X-c_1)g$ for some $g$ that is non-negative over $S$.
\end{itemize}

For Case 1 (resp. Case 2), we can compute a certificate in $\PO(\Nat(S))$ of the
factor $X-c_0$ (resp. $-(X - c_0)$) since
$X - c_0 = (a_0 - c_0) + 1 \cdot (X - a_0)$ (resp.
$-(X - c_0) = (c_0 - b_k) + 1 \cdot (-(X - b_k))$). For Case 3, the construction
relies on \th\ref{sec:cert-elem-ponats}.

From the above properties, a method to construct a certificate of $f$ in
$\Pos(S)$ in terms of natural generators $\Nat(S)$ can iterate the inductive
step on the polynomial $g$ mentioned above until it is a non-negative polynomial
over $S$.  We illustrate this construction with the following example:

\begin{example1}
  \th\label{sec:background-2} Consider the semialgebraic set
  $S_1 := [-1, 1] \cup \{2\} \cup [3, 4]$. The natural generators of $S_1$ are
  $\Nat(S_1) = \{ g_1 : X+1, g_2 : (X-1)(X-2), g_3 : (X-2)(X-3), g_4 :
  -(X-4)\}$. 

  Consider a polynomial $f_1 : = (X + 2)(X-\frac{3}{2})(X-\frac{4}{3})$, which
  is non-negative over $S_1$. Based on the construction suggested by the
  inductive proof in ~\th\ref{sec:background}, $f$ is factored as $(X+2)$ and
  $(X-\frac{3}{2})(X-\frac{4}{3})$ and the certificates of each factor are
 computed.
  
  \begin{itemize}
  \item A certificate $X + 2$ is $1 + 1 \cdot g_1$

\item A certificate of $(X-\frac{3}{2})(X-\frac{4}{3})$ is
    $\frac{1}{6}((X-1)^2 + 1) + \frac{5}{6} \cdot g_2$
  \end{itemize}

  Hence, a certificate of $f$ is
  $\frac{1}{6}((X-1)^2 + 1) + \frac{5}{6} \cdot g_2 + \frac{1}{6}((X-1)^2 + 1)
  \cdot g_1 + \frac{5}{6} \cdot g_1 g_2$.

  Note that while the above certificates of $X+2$ and
  $(X-\frac{3}{2})(X-\frac{4}{3})$ belong to $\QM(\Nat(S_1))$ as well as
  $\PO(\Nat(S_1))$; however, the above certificate of $f$ is only for
  $\PO(\Nat(S_1))$ and not for $\QM(\Nat(S_1))$ since $(X+1)(X-1)(X-2)$ is not a
  generator in $\QM(\Nat(S_1)$.
\end{example1}

The above example illustrates the subtlety that a certificate of a polynomial
$f$ in a preordering structure $\PO(G)$ need not serve as a certificate in a
quadratic module $\QM(G)$, even though $\QM(G) = \PO(G)$. A certificate of a
product of factors is computed by taking the product of certificates of
individual factors in a preordering structure, whereas it need not work in the case
of generating a certificate in the associated quadratic module structure because
a preordering is closed under multiplication. Computing a certificate of a
polynomial in a quadratic module requires generating certificates of products of
natural generators as well.

 \section{Certificate of Products of Natural Generators in a Saturated Quadratic Module} \label{cert_prods}

As noted above, there exists a procedure to compute certificates of a
non-negative polynomial $f \in \qX$ over a compact semialgebraic set $S$ in
$\preord(\Nat(S))$. To obtain a certificate of $f$ in $\quadmod(\Nat(S))$, it is
enough to find certificates of the products of natural generators in
$\quadmod(\Nat(S))$. In this section, all five cases of products of natural
generators are considered one by one. By replacing $X$ with $-X$, these cases
reduce to the following three cases:

\begin{itemize}
\item Case 1: Product of a left linear factor $X - a_0$ and a right linear
  factor $-(X - b_k)$
\item Case 2: Product of a left linear factor $X - a_0$ and a quadratic factor
  $(X - b_i)(X - a_{i+1})$
\item Case 3: Product of two quadratic factors $(X - b_i)(X - a_{i+1})$ and
  $(X - b_j)(X - a_{j+1})$
\end{itemize}

For Case 1, we distinguish two cases. If $a_0 = b_k$, the certificate of the
product is computed using the following identity
$-(X-a_0)^2 = \frac{1}{4}(X - a_0 -1)^2 \cdot (X - a_0) + \frac{1}{4}(X-a_0+1)^2
\cdot (-(X-a_0))$. If $a_0 < b_k$, the certificates are obtained using the
following identity
$(X-a_0)(-(X-b_k)) = \frac{1}{b_k-a_0}((X-b_k)^2 \cdot (X-a_0) + (X-a_0)^2 \cdot
(-(X-b_k)))$.

Certificates in quadratic modules of products of natural generators sharing
common zeros of the form $(X-a_0), (X-a_0)(X-a_1)$ [Case 2] and
$(X-b_i)(X-a_{i+1}), (X-a_{i+1})(X-a_{j+1})$ [Case 3] are already discussed in
the proof of \cite[Theorem 2.5]{marshall2002}. In \cite[Theorem
3.5]{marshall2005}, the authors proved that
$\quadmod(\Nat(S)) = \preord(\Nat(S))$ when $S$ is compact by showing the
membership of the products of $\Nat(S)$ in $\quadmod(\Nat(S))$. Thus, to complete
the remaining cases, we turn the proof of \cite[Theorem 3.5]{marshall2005} into
an algorithm to compute certificates for Case 2 and Case 3 when the natural
generators do not share common zeros. We address each case separately in the
following theorems.

\begin{theorem}
  Let $g_i = X - a_0$ and $g_j = (X - b_i)(X - a_{i+1})$ be the left linear
  factor and a quadratic factor of $\Nat(S)$ such that $a_0 < b_i$. There exists
 a certificate of $g_i g_j$ in $\quadmod(g_i, g_j, -(X - b_k))$ and is
  computable.
\end{theorem}

\begin{proof}
  Choose $\beta \in \mathbb{Q}$ such that $h := -(X - \beta)$ is
  strictly positive over $S$. We have $h \in \quadmod(\Nat(S))$
 since $h = (\beta - b_k) + 1 \cdot (-(X - b_k))$.

  Applying $\algSOSBasicLemma$ with $g_j$ and $g_ih$, we obtain
  $\sigma_1, \tau_1 \in \sum{\aX}^2$ such that $1 = \sigma_1 g_ih + \tau_1
  g_j$. Hence, multiplying $g_ig_jh$ by both sides of the previous equation, we
 obtain

  \begin{equation} 
    \label{sec:cert-prod-natur}
    \begin{split}
      g_i g_j h &= (g_i h)^{2} \sigma_1 \cdot g_j
                  + g_j^{2} \tau_1 \cdot g_i h\\
                &= (g_i h)^{2}
                  \sigma_1 \cdot g_j
                  +g_j^{2} \tau_1 \sigma_2 \cdot g_i
                  + g_j^{2} \tau_1 \tau_2 \cdot h
    \end{split}
  \end{equation}
  by Case 1 for some $\sigma_2, \tau_2 \in \sum{\aX}^{2}$. Notice that
  $\semialgebraic(g_ig_jh) = [a_0, b_i] \cup [a_{i+1}, \beta]$. Hence, we again use
 $\algSOSBasicLemma$ with $g_ig_j$ and $h$ to obtain sums of squares
  $\sigma_3, \tau_3$ such that $1 = \sigma_3 g_ig_j + \tau_3 h$. Hence,
  $g_ig_j = g_i^{2}g_j^{2} \sigma_3 + \tau_3 \cdot g_i g_j h$. Substituting
 the certificate of $g_i g_j h$ from the equation
 \eqref{sec:cert-prod-natur} in the previous equation, we obtain a certificate of $g_ig_j$ in
  $\quadmod(g_i, g_j, h) \subseteq \quadmod(g_i, g_j, -(X - b_k))$.
\end{proof}

\begin{theorem} 
  \th\label{sec:cert-prod-natur-3}
  Let $g_i = (X - b_i)(X - a_{i+1})$ and $g_j = (X - b_j)(X - a_{j+1})$ be two
  different quadratic factors of $\Nat(S)$ such that $a_{i+1} < b_j$. There
 exists a certificate of $g_i g_j$ in $\quadmod(X - a_0, g_i, g_j, -(X - b_k))$
  and is computable.
\end{theorem}

\begin{proof}
  Choose $\alpha, \beta \in \mathbb{Q}$ such that $h_1 := X - \alpha$ and
  $h_2 := -(X - \beta)$ are strictly positive over $S$. Both $h_1$ and $h_2$ are
  members of $\quadmod(\Nat(S))$ since
  $h_1 = (a_0 - \alpha) + 1 \cdot (X - a_0)$ and
  $h_2 = (\beta - b_k) + 1 \cdot (-(X - b_k))$.

  We apply $\algSOSBasicLemma$ with $h_1g_ih_2$ and $g_j$ to obtain
  $\sigma_1, \tau_1 \in \sum{\aX}^{2}$ such that
  $1 = \sigma_1 h_1g_ih_2 + \tau_1 g_j$. Hence, multiplying $h_1 g_i g_j h_2$ by
 both sides of the previous equation, we obtain

  \begin{equation} 
    \label{sec:cert-prod-natur-2}
    \begin{split}
      h_1g_ig_jh_2
      &= \sigma_1 (h_1 g_i h_2)^{2} \cdot g_j
        + \tau_1 g_j^{2} \cdot h_1 g_i h_2\\
      &= \sigma_1 (h_1 g_i h_2)^{2} \cdot g_j
        + \tau_1 g_j^{2} s_1 \cdot h_1\\
      &+ \tau_1 g_j^{2} s_2 \cdot g_i
        + \tau_1 g_j^{2} s_3 \cdot h_2
    \end{split}
  \end{equation}

  where the last equation is obtained from a similar derivation of the equation
 \eqref{sec:cert-prod-natur} for the polynomial $h_1g_ih_2$. Finally, we apply
  $\algSOSBasicLemma$ with $g_ig_j$ and $h_1h_2$ to obtain
  $\sigma_2, \tau_2 \in \sum{\aX}^{2}$ such that
  $1 = \sigma_2 g_ig_j + \tau_2 h_1 h_2$. We obtain a certificate of $g_ig_j$ by
  multiplying $g_ig_j$ by both sides of the previous equation and substitute the
  certificate of $h_1g_ig_jh_2$ in
  $\quadmod(h_1, g_i, g_j, h_2) \subseteq \quadmod(X - a_0, g_i, g_j, -(X -
  b_k))$.
\end{proof}

\begin{example1}
  \th\label{sec:cert-prod-natur-1}
 Consider $G = \{g_1, g_2, g_3, g_4\}$ where $g_1 = X + 3$, $g_2 = (X +
  2)(X + 1)$, $g_3 = (X - 1)(X - 2)$ and $g_4 = -(X - 3)$. We want to compute a
  certificate of $g_2 g_3$ in $\quadmod(G)$. We apply the procedure in the proof
  of \th\ref{sec:cert-prod-natur-3} for the latter:

  First, we compute the sums of squares $\sigma_1, \tau_1$ for the identity
  $1 = \sigma_1 g_1 g_2 g_4 + \tau_1 g_3$. Using the Algorithm
  $\algSOSBasicLemma$ with inputs $g_1 g_2 g_4$ and $g_3$ we obtain: 

  {\scriptsize
    $\sigma_1 = \frac{2240663747533}{63928589899840}(-\frac{2838880455867}{35850619960528}X^{3}-\frac{5070279942435}{35850619960528}X^{2}+X-\frac{1954774155034}{6721991242599})^2\\
    +\frac{4806615684894157212658049}{322295566093820521913463120}(\frac{1965594740382545239529541}{307623403833226061610115136}X^{3}-\frac{17069241747242628077096889}{153811701916613030805057568}X^{2}+1)^2\\
    +\frac{45055618345457224538985176048524046579}{39331860854494354284103736102538215956480}(-\frac{60091193011842907178590284062434677153}{180222473381828898155940704194096186316}X^{3}+X^{2})^2\\
    +\frac{1608598652282431290179570962045004567738313138857}{5898940718879626314976720334559556906740992214317793280}X^6$

    $\tau_1 =
    \frac{3408243953093221}{8182859507179520}(-\frac{81330974687725}{3408243953093221}X^{4}-\frac{3396784247798305}{30674195577838989}X^{3}+\frac{630553521671491}{3718084312465332}X^{2}+X+\frac{1844281814531328}{3408243953093221})^2+\frac{5913932904047613728065258471}{27235528744488255343321623080}(\frac{10740479987679269897853128335}{1703212676365712753682794439648}X^{4}-\frac{260955860135590366578430046149}{72670407524937077490465896091648}X^{3}-\frac{2379845916734658371072298110809}{12111734587489512915077649348608}X^{2}+1)^2+\\\frac{1085456813189269522216530332633425792209444565}{118930347021208299135706252768780883552781729792}(-\frac{1602913928501451029422394636199921184431802688}{16281852197839042833247954989501386883141668475}X^{4}\\
    +\frac{9671074188052417654083129659829142388068001791}{97691113187034256999487729937008321298850010850}X^{3}+X^{2})^2\\+\frac{81483365955960571648187500735366010328103022996482467469550969}{690675253323385404380206338600288558068186323163031057031692288000}(\\\frac{25734754223651000818159963525358900260248917738032343012730648}{81483365955960571648187500735366010328103022996482467469550969}X^{4}+X^{3})^2\\+\frac{13317839023340912696096944737535918990630326698700788196454800010845723911}{6751015224870915008960015389274073922308658036410110816951164361661316488668160}
    X^{8}$
  } 

  Then, using these polynomials, we obtain a certificate of $g_1 g_2 g_3 g_4$ in
  $\quadmod(g_3, g_1 g_2 g_4)$ as
  $g_1 g_2 g_3 g_4 = \sigma_1 (g_1 g_2 g_4)^{2} \cdot g_3 + \tau_1 g_3^{2} \cdot
  g_1 g_2 g_4$. In order to lift this certificate to $\quadmod(G)$, we compute
 two sums of squares $\sigma_2, \tau_2$ using $\algSOSBasicLemma$ with inputs
  $g_1 g_4$ and $g_2$ such that $1 = \sigma_2 g_1 g_4 + \tau_2 g_2$. We obtain:

  $\sigma_2 =
  \frac{5077893}{74983750}(-\frac{1210241}{3385262}X+1)^2
  +\frac{51983202049}{253839639492500}X^2$

  $\tau_2 =
  \frac{29282713}{149967500}(-\frac{22495125}{117130852}X+1)^2
  +\frac{116498980322719}{70263284189240000}X^2$ 

  Hence,
  $g_1 g_2 g_4 = \sigma_2 (g_1 g_4)^{2} \cdot g_2 + \tau_2 g_2^2 \cdot g_1 g_4 =
  \sigma_2 (g_1 g_4)^{2} \cdot g_2 + \frac{1}{6} \tau_2 (g_4 g_2)^2 \cdot g_1 +
  \frac{1}{6} \tau_2 (g_1 g_2)^2 \cdot g_4$. The last equation follows from
  applying Case 1 in this section. Now, we compute the sums of squares
  $\sigma_3, \tau_3$ for the identity $1 = \sigma_3 g_1 g_4 + \tau_3 g_2
  g_3$. Using the Algorithm $\algSOSBasicLemma$ with inputs $g_1 g_4$ and
  $g_2 g_3$ we obtain $\sigma_3 = \frac{1}{40}X^2 + \frac{1}{10}, \tau_3 =
  \frac{1}{40}$ 

  Thus, multiplying $g_2 g_3$ to the identity we obtain
  $g_2 g_3 = \sigma_3 \cdot g_1 g_2 g_3 g_4 + \tau_3 (g_2 g_3)^{2}$. Subtituting
  the certificate of $g_1 g_2 g_3 g_4$ obtained before in the last equation we
  obtain a certificate of $g_2 g_3$ in $\quadmod(G)$.
\end{example1}

 \section{Certificates of Natural Generators in terms of a set of generators $G$} \label{repr_nat_gens}

Although natural generators belong to any univariate saturated quadratic
module, these might not be included in a given set of generators. This section
addresses the problem of computing certificates of natural generators using
elements of the original set of generators.

Factors of the natural generators of a saturated quadratic modules belong to the
factors of the original generators $G$. This follows as a consequence of
\th\ref{sec:cert-elem-ponats-2}. From this observation, we develop a method
based on the Basic Lemma (\th\ref{sec:results-} and \th\ref{sec:results--1}) to
compute certificates of non-negative factors from the given generators. This
method computes certificates for the left natural factor $X - a_0$, right natural
factor $-(X - b_k)$, and quadratic factors $(X - b_i)(X - a_{i+1})$ that belong
to a single generator. When the factors of quadratic factors belong to two
different generators in $G$, we apply a procedure to compute a certificate
involving these two generators.

Motivated by the above construction, in this section we describe a procedure
to compute certificates non-negative factors from the original generators. Then,
we develop a procedure to compute certificates of a quadratic factor
$(X - b_i)(X - a_{i+1})$ using the certificates of the non-negative factors
$(X - b_i)(X - c_{i_1})^{m_{i_1}}$ and $(X - c_{i_2})^{m_{i_2}}(X - a_{i+1})$
from two generators $g_1, g_2 \in G$. Finally, we use these constructions to
provide an algorithm to compute certificates of natural generators in terms of a
given set of generators $G$.
\vspace{-6mm}

\subsection{Constructing certificates of Positive Polynomials}
\label{sec:results--10}

The algorithm proposed in \cite{weifeng2025} to compute the certificate of a
strictly positive polynomial $f$ in Archimedean quadratic modules $\quadmod(G)$
is used whenever certificates are needed for such polynomials. It is based on
a modification and extension of the following result:

\begin{lemma}\cite[Lemma 7][Page 4]{10.1007/s10957-012-0261-9}
  \th\label{lemma_7_averkov}
  Let $G = \{g_1, \dots, g_s\}$, $S := \semialgebraic(G)$, and let $f \in \aX$
  be strictly positive on $S$. Let $B$ be a compact subset of
  $\mathbb{R}^{d}$. Then there exists $g \in \quadmod(G)$ such that $f - g$ is
  strictly positive on $B$.
\end{lemma}

The algorithm for the univariate case combines two ideas: first, it computes $g$
in $\quadmod(G)$ in order to avoid the inclusion of a polynomial of the form
$N - X^{2}$ in the set of generators; second, it uses a generalization of Lemma
7 in \cite{10.1007/s10957-012-0261-9} which produces a polynomial of the form
$f - \tilde{g}$ that is strictly positive over $\mathbb{R}$ where
$\tilde{g} \in \quadmod(G)$. The latter is a sums of squares in the univariate
case. We state the generalization of the Averkov lemma used in this algorithm.

\begin{theorem} \cite[Proposition 5.1]{weifeng2025} Let $f, g \in \aX$ be two
  polynomials with $\semialgebraic(g)$ bounded and $ f > 0 $ on
  $\semialgebraic(g)$. Then there exists a sum of squares $\delta$ such that
  $f - \delta \cdot g$ is strictly positive on $\semialgebraic(g)$ and has a
  lower bound over $\mathbb{R}$.
\end{theorem}

We use the algorithm $\algPutinar$ from \cite{weifeng2025} to compute
certificates of strictly positive polynomials over bounded semialgebraic
sets. This algorithm does not introduce new constants into the base field
$\mathbb{A}$, as intermediate polynomials and parameters use rational
numbers.

 \subsection{Compactification of polynomials over Archimedean quadratic modules}
\label{compactification}

In this section, we introduce the concept of \textbf{compactification} of a
polynomial with respect to the generators $G$ of an Archimedean quadratic module
$\quadmod(G)$ \footnote{This quadratic module does not need to be saturated}. In
simple terms, given a polynomial $f \in \quadmod(G)$, the compactified
polynomial $f$, denoted $\compact_G(f)$, is a polynomial such that
$\semialgebraic(\compact_G(f))$ is bounded and $\compact_G(f)$ satisfies the
same order and sign conditions at the roots of $f$. To accomplish the latter, we
find a polynomial $h$ such that $f h$ satisfies the above properties of
$\compact_G(f)$; consequently $\compact_G(f) \in \quadmod(G)$. We remove the
subscript $G$ in $\compact_G$ whenever it is clear from the context the set of
generators being considered.

If $\semialgebraic(f)$ is bounded, then $\compact(f) = f$; otherwise, there are
two cases to consider: the semialgebraic set of the polynomial $f$ is unbounded
to the right (resp. left), i.e.,
$\semialgebraic(f) = \bigcup_{i=0}^{k}{[a_i, b_i]}$, with $a_0 \in \mathbb{A}$,
and $b_k = \infty$ (resp. $a_0 = -\infty$, and $b_k \in \mathbb{A}$), and when
the semialgebraic set $f$ is unbounded from both sides, i.e., $a_0 = -\infty$
and $b_k = \infty$. For the first case, the polynomial $h$ is a monomial of the
form $-(X - a)$ (resp. $X - a$) when $\semialgebraic(f)$ is unbounded to the
right (resp. left); for the second case, the polynomial $h$ is of the form
$-(X - a)(X - b)$. In the following, we introduce algorithms to compute the
polynomial $h$ and, given certificates of $f$ in $\quadmod(G)$, we can compute a
certificate of $\compact(f)$ in $\quadmod(G)$.

\subsubsection{Auxiliary polynomial $h$ when $\semialgebraic(f)$ is unbounded to
  the right}

The intuitive idea is to lift the polynomial $-f$ by a positive constant
$\epsilon$ such that $\epsilon - f$ only has one real root. The additional roots
of $\epsilon - f$ are complex conjugates, so these are sums of squares for the
multiplier $\tau$. In order to guaranty that the properties that $\epsilon - f$ have
one real root, the positive constant $\epsilon$ is chosen over an interval $I$
such that $f$ increases monotonically outside of $I$; this is done using the
fact that whenever the derivative of a polynomial is positive, the original
polynomial increases monotonically. To compute this interval $I$ we use
bounds encoding the left and right end points of the semialgebraic set
of the generators and the Cauchy bound \cite[Lemma 10.2, page
354]{10.5555/1197095} of the derivative of $f$ to guaranty the
monotonically increasing condition.

\begin{algorithm}
  \caption{Computing auxiliary polynomial $h$ when $\semialgebraic(f)$ is
    unbounded to the right}
  \label{sec:addit-constr-from-4}
  \kwInput{$f \in \aX, bound \in \mathbb{A}$}
  \kwOutput{$\sigma \in \mathbb{A}^{+}$ and $\tau, h \in \aX$}
  \kwRequires{$\deg(f)$ is odd and the leading coefficient of $f$ is positive}
  \kwEnsures{$1 = \sigma f + \tau h$, $\sigma, \tau \in \sum{\aX}^2$}  
  Let $M$ be the Cauchy bound of the derivative of $f$\;
  $I := [-M, \max(bound, M)]$\;
  $\epsilon_1 := \max_{x \in I}(f)$\;
  \tcp{There is only one real root of $\epsilon_1 - f$ of odd degree}
  $\alpha := \text{real root of } \epsilon_1 - f$ of odd degree\;
  $\sigma := \frac{1}{\epsilon_1}$\;
  $h := -(X - \alpha)$\;
  $\tau := \sigma \cdot quotient(\epsilon_1 - f, h)$\;
  \Return $\sigma, \tau, h$\;   
\end{algorithm}

\begin{example1}
  \th\label{sec:addit-constr-from-2} Let us consider
  $g_1 = (X + 1)(X - 2)(X - 4)$ and $g_2 = -(X - 1)$. The semialgebraic set of
  $G = \{g_1, g_2\}$ is the interval $[-1, 1]$ so it is bounded, however
  $\semialgebraic(g_1)$ is $[-1, 2] \cup [4, \infty)$, which it is unbounded to
  the right. We use Algorithm \ref{sec:addit-constr-from-4} to compute the
  polynomial $h$ such that $\compact(f) = fh$. The largest end point of
  $\semialgebraic(G)$ of $1$, which is used as the bound for Algorithm
  \ref{sec:addit-constr-from-4}.

  The Cauchy bound of the derivative of $g$ is $5$, therefore we compute the maximum
  value of $g$ between the interval $[-5, 5]$ to be $\epsilon_1 = 18$. Then we
  factorize $\epsilon_1 - g_1 = (X^{2} + 2)(X - 5)$. Hence, the real root
  $\alpha$ of $\epsilon_1 - g_1$ of odd degree is $5$.

  Following the algorithm, $\sigma := \frac{1}{18}$ is a positive
  constant. We have $h = -(X - 5)$, and finally $\tau
  = \frac{1}{18} (X^{2} + 2)$.
\end{example1}

\subsubsection{Auxiliary polynomial $h$ when $\semialgebraic(f)$ is unbounded in
  both sides}

Similarly to the previous case, we lift the polynomial $-f$ by a positive
constant. Since $f$ is unbounded to both sides, the polynomial $\epsilon - f$
will have only two real roots, and the rest of its roots are complex conjugates.

\begin{algorithm}
  \caption{Computing auxiliary polynomial $h$ when $\semialgebraic(f)$ is
    unbounded from both sides}
  \label{sec:addit-constr-from-5}
  \kwInput{$f \in \aX, lowerbound, upperbound \in \mathbb{A}$}
  \kwOutput{$\sigma \in \mathbb{A}^{+}$ and $\tau, h \in \aX$}
  \kwRequires{$\deg(f)$ is even and leading coefficient of $f$ is positive}
  \kwEnsures{$1 = \sigma f + \tau h$, $\sigma, \tau \in \sum{\aX}^2$}

  \uIf{$f$ is a positive constant}{
    \Return $\frac{1}{f}$, 0, 0\;
  }
  \Else{
    Let $M$ be the Cauchy bound of the derivative of $f$\;
    $I := [\min(lowerbound, -M), \max(upperbound, M)]$\;
    $\epsilon_1 := \max_{x \in I}(f)$\; 
    $\sigma := \frac{1}{\epsilon_1}$\;
    \tcp{There should be only two real roots of $\epsilon_1 - f$ of
      odd degree}
    $\alpha := \set{x \mid x \in \mathbb{A}, \epsilon_1 - f(x) = 0}$\;
    $h := -\prod_{r \in \alpha}^{}{X - r}$\;
    $\tau := \sigma \cdot quotient(\epsilon_1 - f, h)$\;
    \Return $\sigma, \tau, h$\;  
  }
\end{algorithm}

\subsubsection{Compactification algorithm}
\label{sec:comp-algor-2}

We assume that $\semialgebraic(G) = \bigcup_{i=0}^{k}{[a_i, b_i]}$. Using the
algorithms above to compute the auxiliary polynomial $h$ and the sums of squares
$\sigma, \tau$ for the identity $1 = \sigma f + \tau h$, we set the values for
the lowerbounds and upperbounds with $a_0$ and $b_k$ respectively to force $h$
is strictly positive over $\quadmod(G)$. Using the algorithm $\algPutinar$, we
compute certificates of $h$ in $\quadmod(G)$. Finally, we use the following
algorithm to compute the certificates of $\compact(f)$ in $\quadmod(G)$.

\begin{algorithm}
  \DontPrintSemicolon
  \caption{Compactification algorithm}
  \label{sec:comp-algor}
  \kwInput{$f \in \aX$, $G = \{g_1, \dots, g_s\} \subseteq \aX$}
  \kwOutput{$\tau_0, \tau_1, \dots, \tau_s \in \aX$}
  \kwRequires{$f = \sigma_0 + \sum_{i=1}^{s}{\sigma_i \cdot g_i} \in \quadmod(G)$
      and $\semialgebraic(G) = \bigcup_{i=0}^{k}{[a_i, b_i]}$}
  \kwEnsures{$\compact_G(f) = \tau_0 + \sum_{i=1}^{s}{\tau_i \cdot g_i}$,
    $\tau_i \in \sum{\aX}^2$ and $\semialgebraic(\compact_G(f))$ is bounded}
  \uIf{$\semialgebraic(f)$ is bounded}{
    \Return $\sigma_0, \sigma_1, \dots, \sigma_s$\;
  }
  \uIf{$\semialgebraic(f)$ is unbounded to the right}{
    Let $\sigma, \tau, h$ be the output of Algorithm
    \ref{sec:addit-constr-from-4} with input $f$\; 
  }
  \uIf{$\semialgebraic(f)$ is unbounded to the left}{
    Let $\sigma, \tau, h$ be the output of Algorithm
    \ref{sec:addit-constr-from-4} with input $f[X/-X]$\; 
    $\sigma := \sigma[X/-X], \tau := \tau[X/-X], h := h[X/-X]$\;
  }
  \uIf{$\semialgebraic(f)$ is unbounded from both sides}{
    Let $\sigma, \tau, h$ be the output of Algorithm
    \ref{sec:addit-constr-from-5} with input $f$\; 
  }
  Let $\rho_0, \dots, \rho_s$ be the certificate obtained from Algorithm
  $\algPutinar$ with input $G$, $h$\;
  \Return $\sigma f^2 \rho_0 + \tau h^{2} \sigma_0, \dots,
  \sigma f^2 \rho_s + \tau h^{2} \sigma_s$\; 
\end{algorithm}

Since both Algorithm \ref{sec:addit-constr-from-4} and Algorithm
\ref{sec:addit-constr-from-5} satisfy $1 = \sigma f + \tau h$, then multiplying
$\compact(f) = fh$ on both sides of the previous equation gives 
$f h = \sigma f^{2} \cdot h + \tau h^{2} \cdot f$. Since $h$ is strictly
positive over $\semialgebraic(G)$, we can use algorithm $\algPutinar$ to compute
a certificate of it in $\quadmod(G)$, hence the result.

\begin{example1}
  \th\label{sec:comp-algor-1} We will continue with
  \th\ref{sec:addit-constr-from-2}. We need to compute a certificate of $h$ in
  $\quadmod(g_1, g_2)$. Using $\algPutinar$, we obtain
  $h = \frac{78266}{78991} ((X - \frac{19929}{78266})^2 +
  \frac{2949332587}{6125566756}) + \frac{39133}{78991} \cdot g_1 +
  \frac{39133}{78991} (X - 1)^{2} \cdot g_2$. Notice that there is a simpler
  certificate of $h$ in $\quadmod(g_1, g_2)$ using $h = 4 + g_2$ which we will
  use in this example. Hence, a certificate of the compactified generator $g_1$
  is
  $\compact_{\{g_1, g_2\}}(g_1) = 4 \sigma g_1^2 + h^2 \tau \cdot g_1 + \sigma
  g_1^2 \cdot g_2$. We can verify
 $\semialgebraic(\compact_{\{g_1, g_2\}}(g_1)) = [-1, 2] \cup [4, 5]$.
\end{example1}

\begin{example1}
  \th\label{sec:comp-algor-3} The following example shows that Algorithm
  \ref{sec:comp-algor} can effectively handle the case when the associated
  semialgebraic set of the generators is empty. Consider the generators
  $G = \{g_1, g_2, g_3\}$ where $g_1 := X - 1$,
  $g_2 := -(X + 2) (X + 1) (X - 2) (X - 3)$ and $g_3 := (X + 1) (X - 4)$. We
  want to compactify the generator $g_1$. Since $\semialgebraic(g_1)$ is
  unbounded to the right, we use Algorithm \ref{sec:addit-constr-from-4} with
  inputs $g_1$ and $1$ as bound to compute sums of squares $\sigma, \tau$ and an
  auxiliary polynomial $h$. The output of this step is $\sigma = 1, \tau = 1$ and
  $h = -(X - 2)$.  Then, we compute the certificate of $h$ in
  $\{g_1, g_2, g_3\}$ as any polynomial is strictly positive over the empty
  set. For the latter, we use the algorithm $\algPutinar$ to obtain
  $h = \rho_0 + \rho_2 \cdot g_2$ where:

  \begin{itemize}
  \item[] $\rho_0 =
    \frac{1616315}{140311}(-\frac{13168023}{64652600}X^{2}
    +\frac{469529}{3232630}X+1)^2\\
    +\frac{146587437152}{226786773965}
    (-\frac{3674166724833}{11726994972160}X^{2}+X)^2 
    +\frac{1304025475172613633}{658170556615496704000}X^{4}$
  \item[] $\rho_2 = \frac{76230}{140311}$
  \end{itemize}

  Finally, we obtain the certificate of $\compact(g_1) = g_1 h$ by combining the
  certificates of $g_1$ and $h$ in Line 11 from Algorithm \ref{sec:comp-algor},
  we obtain
  $\compact_G(g_1) = g_1^{2} \rho_0 + h^{2} \cdot g_1 + g_1^{2} \rho_2 \cdot
  g_2$.
\end{example1}

 \subsection{Certificates of non-negative factors of a members of $\quadmod(G)$}
\label{sec:cert-non-neg-cert}

Suppose that $f = t_1 t_2 \in \quadmod(G)$ where both $t_1$ and $t_2$ are
coprime and non-negative over $\semialgebraic(f)$ and assume we have
certificates of $f$ in $\quadmod(G)$, i.e.,
$f = \sigma_0 + \sum_{i=1}^{s}{\sigma_i \cdot g_i}$. Our goal is to compute
certificates of $t_1$ in $\quadmod(G)$ using the certificates
$\sigma_0, \dots, \sigma_s$ of $f$. There are two cases that consider the
boundedness of $\semialgebraic(f)$.

In case $\semialgebraic(f)$ is bounded, using algorithm $\algSOSBasicLemma$ with
inputs $t_1$ and $t_2$, we compute sums of squares polynomials $\sigma, \tau$
such that $1 = \sigma t_1 + \tau t_2$. The certificate of $t_1$ is computed by
multiplying $t_1$ to both sides of the previous equation, thus obtaining
$t_1 = t_1^2 \sigma + \tau \cdot f$. Finally, we obtain a certificate of $t_1$
in $\quadmod(G)$ by substituting $f$ with its certificate in $\quadmod(G)$.

In case $\semialgebraic(f)$ is unbounded, we use the compactification algorithm
\ref{sec:comp-algor} to compactify the polynomial $f$. Then, the above
construction for the bounded case is used, thus finding a certificate of the
factor of $t_1$ in $\quadmod(\compact_G(g))$. This certificate is lifted to
$\quadmod(G)$ as the Algorithm \ref{sec:comp-algor} provides a certificate of
$\compact_G(f)$ in $\quadmod(G)$.

 \subsection{Removing redundant intervals from semialgebraic sets}
\label{sec:remov-redund-strictl-9}

In this section, we discuss a method that addresses a shortcoming from Section
\ref{sec:cert-non-neg-cert} about the requirement for both factors of a given
member $f$ in $\quadmod(G)$ to be non-negative over $\semialgebraic(f)$. The
latter limits the application of the Basic Lemma (\th\ref{sec:results-}). We
motivate this construction with the following example:

\begin{example1}
  \th\label{sec:remov-redund-strictl}
 Consider the following set of generators $G = \{g_1, g_2\}$ where
  $g_1 = -(X + 4)(X + 3)(X + 1)(X - 1)(X - 3)(X - 4)$ and
  $g_2 = (X + 2)(X - 2)$. The set of natural generators of $\quadmod(G)$ are
  $n_1 = X + 4, n_2 = (X + 3)(X - 3)$ and $n_3 = -(X - 4)$. Certificates of the
  natural generators $n_1$ (resp $n_3$) can be obtained using the method
  discussed in Section \ref{sec:cert-non-neg-cert} using $g_1$ as these
  polynomials are non-negative over $\semialgebraic(\compact_G(g_1))$. However,
  this method is not applicable to the natural generator $n_2$ as it is not
  non-negative over $\semialgebraic(\compact_G(g_1))$ nor over
  $\semialgebraic(\compact_G(g_2))$.
\end{example1}

In general, given a polynomial $f \in \quadmod(G)$ such that

\begin{equation} 
  \label{sec:remov-redund-strictl-8}
  \begin{split}
    \semialgebraic(f) &= \bigcup_{i=0}^{j}{[c_i, d_i]} \cup
                        \bigcup_{i=j+1}^{k}{[c_i, d_i]} \cup
                        \bigcup_{i=k+1}^{l}{[c_i, d_i]}\\ 
    \semialgebraic(G) &\cap \bigcup_{i=j+1}^{k}{[c_i, d_i]} = \emptyset
  \end{split}
\end{equation}

we want to find a polynomial $p \in \quadmod(G)$ such that
$\semialgebraic(f + p) \cap (d_j, c_{k+1}) = \emptyset$ and $f + p$ has the same
order and sign conditions at $d_{j}, c_{k+1}$ as $f$. Clearly, if certificates
of $f \in \quadmod(G)$ are given, then it is enough to find certificates of
$p \in \quadmod(G)$ to have a certificate of $f + p \in \quadmod(G)$. In this
way, we can find a certificate of the natural generator $n_2$ in
\th\ref{sec:remov-redund-strictl} by removing the interval $[-1, 1]$ from the
semialgebraic set of $g_1$ and applying the method in Section
\ref{sec:cert-non-neg-cert}.

We describe a procedure to compute the certificates of $f + p$ mentioned above. We
assume $f = \sigma_0 + \sum_{i=1}^{s}{\sigma_i \cdot g_i}$. A high level
description is the following:

\begin{enumerate}
\item Compute positive constants $\alpha, \epsilon_1$, and $\epsilon_2$
  satisfying the following conditions:
  \begin{itemize}
  \item $d_j + \epsilon_1 < c_{j+1}$ and $d_k < c_{k+1} - \epsilon_2$
  \item Let
    $h = (X - (d_j + \epsilon_1))(X - (c_{k+1} - \epsilon_2))$,
    $q = (X - d_j)^{\gamma_1}(X - c_{k+1})^{\gamma_2}$ where
    \begin{equation} 
      \label{item:let-h-=}
      \begin{split}
        \gamma_1 &= \begin{cases}
          \ord_{d_j}(f) + 1  & \text{if } \ord_{d_j}(f) \text{ is odd}\\
          \ord_{d_j}(f) + 2 & \text{otherwise} \\
        \end{cases}\\
        \gamma_2 &= \begin{cases}
          \ord_{c_{k+1}}(f) + 1  & \text{if } \ord_{c_{k+1}}(f) \text{ is odd}\\
          \ord_{c_{k+1}}(f) + 2 & \text{otherwise} \\
        \end{cases}
      \end{split}
    \end{equation}
  \item
    $\semialgebraic(f + \alpha q \cdot h) \cap (d_j, c_{k+1}) = \emptyset$
  \end{itemize} 
\item Since $h$ is strictly positive over $\semialgebraic(G)$, compute a
 certificate $\tau_0 + \sum_{i=1}^{s}{\tau_i \cdot g_i}$ of $h$ using the
  $\algPutinar$ algorithm.
\item Return certificate
  $\sigma_0 + \alpha q \tau_0, \dots, \sigma_s + \alpha q \tau_s$.
\end{enumerate}

In the next proposition, we prove that the polynomial $\alpha q \cdot h$ above
satisfies the required order and sign conditions at $d_j, c_{k+1}$.

\begin{proposition}
  \th\label{sec:remov-redund-strictl-2} Let $f$ and $p := \alpha q \cdot h$ be
  the polynomials as described above. We have that,
  \begin{itemize}
  \item $\ord_{d_j}(f + p) = \ord_{d_j}(f)$,
    $\ord_{c_{k+1}}(f + p) = \ord_{c_{k+1}}(f)$
  \item $\epsilon_{d_j}(f + p) = \epsilon_{d_j}(f)$, $\epsilon_{c_{k+1}}(f + p)
    = \epsilon_{c_{k+1}}(f)$
  \end{itemize}
\end{proposition}

\begin{proof}
  Since $\ord_{d_j}(f) < \ord_{d_j}(p)$ and
  $\ord_{c_{k+1}}(f) < \ord_{c_{k+1}}(p)$, the polynomial $f + p$ maintains the same
  sign and order values in $d_j$ and $c_{k+1}$. Thus, the statement holds.
\end{proof}

\begin{example1}
  \th\label{sec:remov-disj-interv-1}

  Continuing with \th\ref{sec:remov-redund-strictl}, we want to remove the
  redundant interval $[-1, 1]$ from the semialgebraic set of $g_1$. We find that
  $\alpha = \frac{175}{216}$ and $\epsilon_1 = \epsilon_2 = 1$ satisfy the
  conditions mentioned above. We have

  \begin{equation*} 
    \label{sec:remov-redund-strictl-1}
    \begin{split}
      &g_1 + \frac{175}{216}((X + 3)(X - 3))^{2} \cdot (X + 2)(X - 2)\\
      &= g_1 + \frac{175}{216}((X + 3)(X - 3))^{2} \cdot g_2\\
      &= \frac{41 X^2 + 79}{216} \cdot (-(X + 6)(X + 3)(X - 3)(X - 6))\\
    \end{split}
  \end{equation*}

  The semialgebraic set of the last polynomial is $[-6, -3] \cup [3, 6]$. The
  certificate in $\quadmod(g_1, g_2)$ of the above expression is straightforward
  to compute as the strictly positive polynomial $h$ used is $g_2$. Hence, we
  can apply Section \ref{sec:cert-non-neg-cert} to compute a certificate of
  $(X + 3)(X - 3)$ in $\quadmod(g_1, g_2)$.
\end{example1}

\begin{example1}
  \th\label{sec:remov-disj-interv-2} In this example, we modify
  \th\ref{sec:remov-redund-strictl} by changing the multiplicities of $X+3$ and
  $X-3$ to be $g_1 = -(X+4)(X+3)^{3}(X+1)(X-1)(X-3)^{5}(X-4)$. The goal is also
  to remove the interval $[-1, 1]$ from the semialgebraic set of $g_1$. Using
  the same parameters $\alpha, \epsilon_1, \epsilon_2$ as in
  \th\ref{sec:remov-disj-interv-1}, we have

  \begin{equation*} 
    \label{sec:remov-redund-strictl-3}
    \begin{split}
      &g_1 + \frac{175}{216}((X + 3)^{2}(X - 3)^{3})^2 \cdot (X + 2)(X - 2)\\
      &= \frac{41 X^2 + 79}{216} \cdot (-(X + 6)(X + 3)^{3}(X - 3)^{5}(X - 6))\\
    \end{split}
  \end{equation*}

  The semialgebraic set of the last polynomial is $[-6, -3] \cup [3, 6]$. 
\end{example1}

\begin{example1}
  \th\label{sec:remov-disj-interv-3}

 Consider $g_1 = -(X+3)^{2}(X+1)(X-1)(X-3)^{4}$ and
  $g_2 = (X + 2)(X - 2)$. We want to remove the redundant interval $[-1, 1]$
  from $\semialgebraic(g_1)$. Using the parameters
  $\alpha = \frac{5}{196}, \epsilon_1 = 1, \epsilon_2 = 1$, we have

  \begin{equation*} 
    \label{sec:remov-redund-strictl-4}
    \begin{split}
      &g_1 + \frac{5}{196}((X + 3)^{2}(X - 3)^{3})^{2} \cdot (X + 2)(X - 2)\\
      &= g_1 + \frac{5}{196}((X + 3)^{2}(X - 3)^{3})^{2} \cdot g_2\\
      &=
        \frac{89}{196}
        \left((-\frac{793}{3560}X^{2}+1)^2+\frac{193}{20}(X)^2
        +\frac{83151}{142400}(X^{2})^2\right)\\
      &\cdot (-(X + 4)(X + 3)^{2}(X - 
        3)^{4}(X - 4))\\ 
    \end{split}
  \end{equation*}

  The semialgebraic set of the last polynomial is $[-4, -3] \cup [3, 4]$.
\end{example1}

The steps 2 and 3 in the high level description are easy to implement while the
first step requires more detail. In order to find values for
$\alpha, \epsilon_1, \epsilon_2$, first we compute a value for $\alpha$ fixing
the values of $\epsilon_1, \epsilon_2$ with
$\frac{c_{j+1} - d_j}{2}, \frac{c_{k+1} - d_k}{2}$, respectively, such that
$\semialgebraic(\alpha h) \cap \bigcup_{i=j+1}^{k}{[c_i, d_i]}$ is empty. The
polynomial $p = \alpha q \cdot h$ introduces additional roots at $d_j$ and
$c_{k+1}$, so the semialgebraic set of $\semialgebraic(f + p)$ might have
intervals near these endpoints; to fix the latter, we update the values of
$\epsilon_1, \epsilon_2$ from
$(0, \frac{c_{j+1} - d_j}{2}), (0, \frac{c_{k+1} - d_k}{2})$, respectively,
 to remove these intervals. We use the method in the following proof to
find a value for $\alpha$:

\begin{proposition}
  \th\label{sec:remov-redund-strictl-7} Let $f \in \quadmod(G)$ with
  semialgebraic set as in \eqref{sec:remov-redund-strictl-8}. Let
  $p = (X - d_j)^{\gamma_1}(X - c_{k+1})^{\gamma_2} \cdot (X - \frac{c_{j+1} -
    d_j}{2})(X - \frac{c_{k+1} - d_k}{2})$ where $\gamma_1, \gamma_2$ are
  defined in \eqref{item:let-h-=}. There exists a positive $\alpha$ such that
  $\semialgebraic(f + \alpha \cdot p) \cap \bigcup_{i=j+1}^{k}{[c_i, d_i]}$ is
  empty.
\end{proposition}

\begin{proof}
  Let $S = \bigcup_{i=j+1}^{k}{[c_i, d_i]}$. The semialgebraic set of $p$ is
  $(d_j, \frac{c_{j+1} - d_j}{2}] \cup [\frac{c_{k+1} - d_k}{2}, c_{k+1})$,
 and therefore $p$ is strictly negative over $S$. For each of the intervals
  $[c_i, d_i]$ in $S$, we compute the maximum value of $p$ and the maximum value of
  $f$, i.e.,
 $\alpha_i = -\frac{\max_{x \in [c_i, d_i]}(f)}{\max_{x \in [c_i, d_i]}(p)} >
  0$. We will prove that $\semialgebraic(f + (\alpha_i + 1) \cdot p) \cap [c_i, d_i]$
 is empty. In contrast, there exists a point $x \in [c_i, d_i]$ such that
  $f(x) + (\alpha_i + 1) \cdot p(x) \geq 0$, which implies
  $f(x) + p(x) \geq - \alpha_i \cdot p(x) = \max_{x \in [c_i, d_i]}(f)
  \frac{p(x)}{\max_{x \in [c_i, d_i]}(p)}$. Notice that
  $\frac{p(x)}{\max_{x \in [c_i, d_i]}(p)} \geq 1$ as $p$ is negative over
  $[c_i, d_i]$ so $f(x) + p(x) \geq \max_{x \in [c_i, d_i]}(f)$, which is a
  contradiction. In fact, any $\beta > \alpha_i$ guaranties that
  $\semialgebraic(f + \beta \cdot p) \cap [c_i, d_i]$ is empty. From this,
  choose $\alpha := 1 + \max\{\alpha_{j+1}, \dots, \alpha_k\}$, which satisfies
  the property of empty intersection with each interval in $S$.
\end{proof}

\begin{proposition}
  \th\label{sec:remov-redund-strictl-5}

  Let the polynomial $f$ and the positive constant $\alpha$ be as in
  \th\ref{sec:remov-redund-strictl-7}. Let
  $p_1 = \alpha (X - d_j)^{\gamma_1}(X - c_{k+1})^{\gamma_2} \cdot (X - (d_j +
  \epsilon_1))(X - (c_{k+1} - \epsilon_2))$ where $\gamma_1, \gamma_2$ are
  defined in \eqref{item:let-h-=}. There exist
  $\epsilon_1 \in (0, \frac{c_{j+1} - d_j}{2})$ and
  $\epsilon_2 \in (0, \frac{c_{k+1} - d_k}{2})$ such that the semialgebraic set
  of $f + p_1$ intersecting $(d_j, c_{k+1})$ is empty.
\end{proposition}

\begin{proof} 
  Let $p$ be from \th\ref{sec:remov-redund-strictl-7}, and
  $t := \alpha \cdot p$. First, we prove that for every
  $x \in [\frac{c_{j+1} - d_j}{2}, \frac{c_{k+1} - d_k}{2}]$ we have 
$p_1(x) \leq t(x)$. By continuity, it suffices to prove the existence of a
  point $x \in [\frac{c_{j+1} - d_j}{2}, \frac{c_{k+1} - d_k}{2}]$ such that
 $p_1(x) \leq t(x)$ since both $p_1$ and $t$ have no additional roots in
  $[\frac{c_{j+1} - d_j}{2}, \frac{c_{k+1} - d_k}{2}]$. This is satisfied with
  $x = \frac{c_{j+1} - d_j}{2}$ as $t(x)$ is zero while $p_1(x)$ is a negative
  value. This implies that the polynomial $p_1$ also satisfies the property
  $\semialgebraic(f + p_1) \cap [\frac{c_{j+1} - d_j}{2}, \frac{c_{k+1} -
    d_k}{2}]$ is empty.

  Now, we prove the existence of $\epsilon_1$ and $\epsilon_2$ satisfying the
  main statement. Suppose, on the contrary, that for every
  $\epsilon_1 \in (0, \frac{c_{j+1} - d_j}{2})$ and
  $\epsilon_2 \in (0, \frac{c_{k+1} - d_k}{2})$, we have 
$\semialgebraic(f + p_1) \cap (d_j, c_{j+1})$ is non-empty. Let
  $I_1 = (d_j, d_j + \epsilon_1)$, $I_2 = (c_{k+1} - \epsilon_2, c_{k+1})$, and
  let us choose a value of $\epsilon_1, \epsilon_2$ close to $0$ such that
 $f$ strictly decreases over $I_1$ and strictly increases over
  $I_2$. Clearly, $f + p_1$ is non-positive over
  $[d_j + \epsilon_1, \frac{c_{j+1} - d_j}{2}]$ and non-positive over
  $[\frac{c_{k+1} - d_k}{2}, c_{k+1} - \epsilon_2]$, so we will prove that
  $\semialgebraic(f + p_1) \cap (I_1 \cup I_2)$ is also empty, thus reaching a
  contradiction.

  We prove if $\semialgebraic(f + p_1) \cap I_1$ is non-empty, then it has at
  most one single interval. If $\semialgebraic(f + p_1) \cap I_1$ has more than
  a single interval, it means that $f + p_1$ should have at least two local
  maxima and one local minimum. The latter implies that the derivative of
  $f + p_1$ has at least three real zeros over $I_1$. However, since $p_1$ has a
  single local maximum in $I_1$, its derivative has a single real zero in
  $I_1$. Additionally, $f$ is strictly decreasing over $I_1$, so the derivative
  of $f$ is negative over $I_1$, therefore the derivative of $f + p_1$ can have at
  most two real zeros in $I_1$. Therefore, if $\semialgebraic(f + p_1) \cap I_1$
  is not empty, then it is a single interval. Furthermore, the left endpoint of
  this interval is $d_j$ as it is a common zero of $f$ and $p_1$. A similar
  argument concludes that if $\semialgebraic(f + p_1) \cap I_2$ is non-empty,
  then it has at most a single interval with right endpoint
  $c_{k+1}$.

  Now, we will prove that if $\semialgebraic(f + p_1) \cap I_1$ is a single
  interval with left endpoint $d_j$, we reach a contradiction, proving our main
  assertion about the existence of a value of $\epsilon_1$. First, we compute
  the Taylor series expansion of $f$ and $p_1$ at $d_j$, i.e.,
  $f = \beta_1 (X - d_j)^{\ord_{d_j}(f)} + \dots$ and
  $p_1 = \alpha \gamma_1 (X - d_j)^{\ord_{d_j}(f) + r} + \dots$ where
  $\beta_1 < 0$ and $r$ is either 1 or 2. Then
  $f + p_1 = (\beta_1 + (X - d_j)^{r}(\dots))(X - d_j)^{\ord_{d_j}(f)}$. If we
  pick a point $\delta$ in $I_1$ close to $d_j$, then the sign of
  $f(\delta) + p_1(\delta)$ is the sign of $\beta_1$, which is negative, as
  $(\delta - d_j)^{r}$ is close to $0$ as well as
  $(\delta - d_j)^{\ord_{d_j}(f)} > 0$. However, this contradicts the fact that
  $\delta \in \semialgebraic(f + p_1) \cap I_1$. A similar argument also shows that
 $\semialgebraic(f + p_1) \cap I_2$ is empty. Therefore, there exist
  values of $\epsilon_1 \in (0, \frac{c_{j+1} - d_j}{2})$ and
  $\epsilon_2 \in (0, \frac{c_{k+1} - d_k}{2})$ such that
  $\semialgebraic(f + p_1) \cap (d_j, c_{k+1})$ is empty.
\end{proof}

The following algorithm allows us to compute the parameters $\alpha, \epsilon_1,
\epsilon_2$ required by the construction.

\begin{algorithm}[ht]
  \DontPrintSemicolon
  \caption{Removing redundant strictly positive factors \\ $ h, \alpha q :=
    \algRemoveStrictPosBet(f) $}
  \label{sec:remov-redund-strictl-6}
  \kwInput{$f \in \qX, d_j, c_{k+1} \in \mathbb{A}$}
  \kwOutput{$h, \alpha q \in \aX$}
  \kwRequires{$\semialgebraic(f) = \bigcup_{i=0}^{j}{[c_i, d_i]} \cup
    \bigcup_{i=j+1}^{k}{[c_i, d_i]} \cup \bigcup_{i=k+1}^{l}{[c_i, d_i]}$}
  \kwEnsures{$\semialgebraic(f + \alpha q \cdot h) \cap (d_j, c_{k+1})$ is
    empty} 

  $\epsilon_1 := \frac{c_{j+1} - d_j}{2}$\;
  $\epsilon_2 := \frac{c_{k+1} - d_k}{2}$\;
  $h := (X - (d_j + \epsilon_1))(X - (c_{k+1} - \epsilon_2))$\; 
  $q := (X - d_j)^{\gamma_1}(X - c_{k+1})^{\gamma_2}$ where $\gamma_1, \gamma_2$
  are defined as in equations \eqref{item:let-h-=}\;
  $\alpha := 1$\;

  \For{$i = j+1$ to $k$}{
    Let $\alpha_i := -\frac{\max_{x \in [c_i, d_i]}(f)}{\max_{x \in
        [c_i, d_i]}(q \cdot h)} $\; 
    \uIf{$\alpha_i > \alpha$}{
      $\alpha := \alpha_i$\;
    }
  }

  $\alpha := 1 + \alpha$\;

  \While{$\semialgebraic(f + \alpha q \cdot h) \cap ((d_j, d_j + \epsilon_1) \cup
    (c_{k+1} - \epsilon_2, c_{k+1})) = \emptyset$}{
    $\epsilon_1 := \frac{1}{2} \epsilon_1$\;
    $\epsilon_2 := \frac{1}{2} \epsilon_2$\;
  }
\Return $h, \alpha q$\;
\end{algorithm}

Finally, we consider another variation of the problem, where we want to remove
the intervals from the semialgebraic set of $f$ located to the left of the
semialgebraic set of $\semialgebraic(G)$ \footnote{The approach is symmetric
  in the sense that it can also address the case when the intervals to be removed
 are on the right of $\semialgebraic(G)$}.
\vspace{-6mm}

\begin{equation} 
  \label{eq:beginspl-semi-=}
  \begin{split}
    \semialgebraic(f) &= \bigcup_{i=0}^{j}{[c_i, d_i]} \cup
                        \bigcup_{i=j+1}^{k}{[c_i, d_i]}\\
    \semialgebraic(G) &\cap \bigcup_{i=0}^{j}{[c_i, d_i]} = \emptyset
  \end{split}
\end{equation}
\vspace{-4mm}

The theoretical approach is the same as the previous problem. The following
algorithm solves this problem.

\begin{algorithm}[ht]
  \DontPrintSemicolon
  \caption{Removing redundant strictly positive factors \\ $ h, \alpha q :=
    \algRemoveStrictPosLeft(f) $}
  \label{sec:remov-redund-strictl-10}
  \kwInput{$f \in \qX, c_{j+1} \in \mathbb{A}$}
  \kwOutput{$h, \alpha q \in \aX$}
  \kwRequires{$\semialgebraic(f) = \bigcup_{i=0}^{j}{[c_i, d_i]} \cup
    \bigcup_{i=j+1}^{k}{[c_i, d_i]}$}
  \kwEnsures{$\semialgebraic(f + \alpha q \cdot h) \cap (-\infty, c_{j+1})$ is
    empty} 

  $\epsilon := \frac{c_{j+1} - d_j}{2}$\;
  $h := X - (c_{j+1} - \epsilon)$\; 
  \uIf{$\ord_{c_{j+1}}(f)$ is even}{
    $q := (X - c_{j+1})^{\ord_{c_{j+1}}(f) + 2}$\;
  }
  \Else{
    $q := (X - c_{j+1})^{\ord_{c_{j+1}}(f) + 1}$\;
  }
  $\alpha := 1$\;
  \For{$i = 1$ to $j$}{
    Let $\alpha_i := -\frac{\max_{x \in [c_i, d_i]}(f)}{\max_{x \in
        [c_i, d_i]}(q \cdot h)} $\; 
    \uIf{$\alpha_i > \alpha$}{
      $\alpha := \alpha_i$\;
    }
  }

  $\alpha := 1 + \alpha$\;

  \While{$\semialgebraic(f + \alpha q \cdot h) \cap (c_{k+1} - \epsilon,
    c_{k+1}) = \emptyset$}{ 
    $\epsilon := \frac{1}{2} \epsilon$\;
  }
  \Return $h, \alpha q$\;
\end{algorithm}

\begin{example1}
  \th\label{sec:remov-redund-strictl-11} Let us consider $g_1 = (X+4)(X+3)(X+1)$
  and $g_2 = -(X + 2)(X - 1)$. The semialgebraic set of
  $\semialgebraic(g_1) = [-4, -3] \cup [1, \infty)$. We want to remove the
  redundant interval $[-4, -3]$ from $\semialgebraic(g_1)$.

  First, we fix $\epsilon = 1$ and set $h = X + 2, q = (X + 1)^{2}$. Then, we
  compute $\alpha$ by taking the ratio of maximum values of $g_1$ and
  $q \cdot h$ over $[-4, -3]$. This is
  $\alpha = 1 - \frac{\max_{x \in [-4, -3](f)}}{\max_{x \in [-4, -3](q \cdot
      h)}} = 1 + \frac{1}{54}(10 - 7 \sqrt{7})$. We can check that
  $f + \alpha q \cdot h = \frac{1}{54} (X + 1) (776 - 14 \sqrt{7} + (570 - 21
  \sqrt{7})X + (118 - 7 \sqrt{7})X^2)$ and
  $\semialgebraic(f + \alpha q \cdot h)$ is $[-1, \infty)$.
\end{example1}

 \subsection{Certificates of natural generators}
\label{sec:cert-natur-gener-1982}

\subsubsection{Certificates of a linear factor}
\label{sec:cert-line-fact}

We compute a certificate of the left linear factor $X - a_0$. Using
\th\ref{sec:cert-elem-ponats-2}, we know that there exists a polynomial
$g \in G$ such that $g(a_0) = 0$ and $\frac{d g}{dX}(a_0) > 0$, that is,
$g = (X - a_0) \tilde{g}$ with $\tilde{g}$ is not divisible by $X -
a_0$.

We can assume $X - a_0$ is non-negative over $\semialgebraic(g)$, as otherwise
we use the algorithm $\algRemoveStrictPosLeft$ in Section
\ref{sec:remov-redund-strictl-9} to find a polynomial $\hat{g}$ that preserves
the order and sign condition at $a_0$ and compute its certificates in
$\quadmod(G)$ such that $X - a_0$ is non-negative over
$\semialgebraic(\hat{g})$. The procedure discussed in Section
\ref{sec:cert-non-neg-cert} allows us to compute a certificate of $X - a_0$
setting $t_{1} = X - a_0$ and $t_{2} = \frac{g}{t_{1}}$ in terms of
$\quadmod(G)$. Similarly, the same procedure is used to compute a certificate
for the right linear factor $-(X - b_k)$.

 \subsubsection{Certificates of a quadratic factor}

First, we state the following lemma that we will use to describe a method to
compute the certificates of quadratic factors $(X - b_i)(X - a_{i+1})$.

\begin{lemma}
  \th\label{sec:splitting-step-5} Let
  $a_0, b_i, c_{i_1}, c_{i_2}, a_{i+1}, b_k \in \mathbb{A}$ such that
  $a_0 \leq b_i < c_{i_2} < c_{i_1} < a_{i+1} \leq b_k$ and $m_{i_1}, m_{i_2}$
  odd numbers. We have
  $(X - b_i)(X - a_{i+1}) \in Q := \quadmod(X-a_0, (X - b_i)(X -
  c_{i_1})^{m_{i_1}}, (X - c_{i_2})^{m_{i_2}}(X- a_{i+1}), -(X - b_k))$ and its
  certificates in $Q$ module are computable.
\end{lemma}

See the proof in the Appendix.

\begin{lemma}
  \th\label{sec:cert-quadr-fact} Let
  $a_0, b_i, c_{i_1}, c_{i_2}, a_{i+1}, b_k \in \mathbb{A}$ such that
  $a_0 \leq b_i < c_{i_1} \leq c_{i_2} < a_{i+1} \leq b_k$ and
  $m_{i_1}, m_{i_2}$ odd numbers. We have
  $(X - b_i)(X - a_{i+1}) \in Q := \quadmod(X-a_0, (X - b_i)(X -
  c_{i_1})^{m_{i_1}}, (X - c_{i_2})^{m_{i_2}}(X- a_{i+1}), -(X - b_k))$ and its
  certificates in $Q$ are computable.
\end{lemma}

See the proof in the Appendix.

By \th\ref{sec:cert-elem-ponats-2}, we know that there exists a polynomial
$g_{i_1} \in G$ such that $g_{i_1}(b_i) = 0$ and
$\frac{d g_{i_1}}{dX}(b_i) < 0$, i.e., $g_{i_1} = (X - b_i) \tilde{g_{i_1}}$
where $\tilde{g_{i_1}}$ is not divisible by $X - b_i$. If $g_{i_1}(a_{i+1}) = 0$
and $\frac{d g_{i_1}}{dX}(a_{i+1}) > 0$ then
$g_{i_1} = (X - b_i)(X - a_{i+1}) t$ and $t$ is not divisible by $X -
a_{i+1}$. We can assume $(X - b_i)(X - a_{i+1})$ is non-negative over
$\semialgebraic(g_{i_1})$, as otherwise we use the algorithm
$\algRemoveStrictPosBet$ to find a polynomial that preserves the sign and order
conditions in $b_i$ and $a_{i+1}$. Hence, we use the procedure discussed in
Section \ref{sec:cert-non-neg-cert} to compute a certificate of
$(X - b_i)(X - a_{i+1})$ by setting $t_{1} = (X - b_i)(X - a_{i+1})$ and
$t_{2} = \frac{g_{i_1}}{t_{1}}$.

If $g_{i_1}(a_i) \neq 0$ or $\frac{d g_{i_1}}{dX}(a_{i+1}) = 0$ then let
$c_{i_1}$ be the real root of $g_{i_1}$ closest to the right, \footnote{i.e.,
 $c_{i_1} - b_i$ is minimal and positive.} of $b_i$ and
$m_{i_1} := \ord_{c_{i_1}}(g_{i_1})$. We find that
$(X - b_i)(X - c_{i_1})^{m_{i_1}}$ divides $g_{i_1}$ and
$b_i < c_{i_1} \leq a_{i+1}$. If $m_{i_1}$ is even, we use
\th\ref{sec:addit-constr-from-3} to compute a certificates of
$(X - b_i)(X - c_{i_1})^{m_{i_1} + 1}$ in
$\quadmod(\compact_G(g_{i_1})) \subseteq \quadmod(G)$. On the other hand, if
$m_{i_1}$ is odd, we use the procedure
discussed in Section \ref{sec:cert-non-neg-cert}, we obtain a certificate of
$(X - b_i)(X - c_{i_1})^{m_{i_1}}$ in
$\quadmod(\compact_G(g_{i_1})) \subseteq \quadmod(G)$.

Similarly, by \th\ref{sec:cert-elem-ponats-2}, there exists a polynomial
$g_{i_2} \in G$ such that $g_{i_2}(a_{i+1}) = 0$ and
$\frac{d g_{i_2}}{d X}(a_{i+1}) > 0$. Repeating the steps above to $g_{i_2}$ we
find certificates of $(X - c_{i_2})^{m_{i_2}}(X - a_{i+1})$ in $\quadmod(G)$
where $c_{i_2}$ is the real root of $g_{i_2}$ closest to the left of $a_{i+1}$
and $m_{i_2}$ is odd.

Finally, if $c_{i_2} < c_{i_1}$ we use \th\ref{sec:splitting-step-5} to obtain a
certificate of $(X - b_i)(X - a_{i+1})$ in
$\quadmod(X - a_0, g_{i_1}, g_{i_2}, -(X - b_k))$. Otherwise, we use
\th\ref{sec:cert-quadr-fact} to obtain a certificate of
$(X - b_i)(X - a_{i+1})$ in $\quadmod(X - a_0, g_{i_1}, g_{i_2}, -(X - b_k))$.

 \subsubsection{Certificates of $-1$}

If the semialgebraic set associated with the set of generators is empty, by
definition the set of natural generators is $\{-1\}$. As the algorithm 
$\algPutinar$ excludes the case when $\semialgebraic(G)$ is empty, we modify the
core result in \th\ref{lemma_7_averkov} to compute a certificate of $-1$ to use
this construction.

In the proof of \th\ref{lemma_7_averkov}, the author proves the non-negativity
of $f - g$ over $B$ for a particular $g \in \quadmod(G)$ constructed in the
proof considering two cases, on $B \cap S_1$ and $B \cap S_1^{c}$ where
$S_1 := \{x \in \mathbb{R} \mid g_i(x) + 2\epsilon \geq 0 \text{ for } g_i \in
G\}$. For the former case, the author uses the constant
$\mu := \min\{f(x) \mid x \in B \cap S_1 \}$. If $B$ is empty, then the
certificate is trivially a sequence of $0$'s as sums of squares
multipliers. However, if both $B$ and $S_1$ are non-empty, the construction
requires that $\mu$ is positive; nonetheless, for the case of an input
polynomial to be $-1$, this cannot be the case. To address this issue, we find an
$\epsilon > 0$ such that $S_1$ is empty. In doing so, the proof of non-negativity
of $f - g$ for the $g$ constructed only relies on the second case. We construct
such $\epsilon$ using the following proposition:

\begin{proposition}
 Let $G = \{g_1, \dots, g_s\}$ be such that $\semialgebraic(G)$ is empty. There
  exists $\epsilon > 0$ such that
  $\{x \in \mathbb{R} \mid g_i(x) + \epsilon \geq 0, g_i \in G\}$ is empty.
\end{proposition}

\begin{proof}
 Consider the case where $G$ is a singleton. In this case, it is enough
  to find the maximum value $v$ of $g_1 \in G$ over $\mathbb{R}$. Then set
  $\epsilon := -\frac{v}{2}$.

  When $G$ has more than one generator and contains a generator $g$ such that
  $\semialgebraic(g)$ is empty, we apply the previous
  construction. Otherwise, we apply the following construction. Let
  $s(x) := \min_{g_i \in G}(g_i(x))$. As the $g_i$'s are polynomials, the
  function $s$ is continuous. Since $\semialgebraic(G)$ is empty, we see that
 $\semialgebraic(s)$ is also empty. Hence, we find the maximum value $v$ of $s$
  over $\mathbb{R}$. Finally, we choose $\epsilon := -\frac{v}{2}$ as in the
  previous case. We prove that this choice of $\epsilon$ makes
  $\{x \in \mathbb{R} \mid g_i(x) + \epsilon \geq 0, g_i \in G\}$
  empty. Otherwise, there exists a point $x^{*}$ such that all $g_i \in G$
  satisfy $g_i(x^{*}) + \epsilon \geq 0$. In particular, let $g \in G$ be the
  polynomial such that $s(x^{*}) = g(x^{*})$, then
  $0 \leq g(x^{*}) - \frac{v}{2} = g(x^{*}) - \frac{\max_{x \in
      \mathbb{R}}(s(x))}{2} \leq g(x^{*}) - \frac{s(x^{*})}{2} \leq
  \frac{g(x^{*})}{2}$, which implies that $g(x^{*}) \geq 0$. As $g$ evaluates to
  a minimum value at $x^{*}$ then for every $g_i \in G$ we have $g_i(x^{*}) \geq
  0$. This means that $x^{*} \in \semialgebraic(G)$, which contradicts the
  assumption that $\semialgebraic(G)$ is empty.
\end{proof}

\begin{example1}
  \th\label{sec:certificates--1} Let $G := \{g_1, g_2\}$ where
  $g_1 := (X + 2)(X + 1)(X - 3)$ and $g_2 := -(X + 3)(X - 1)(X - 2)$. We can
  check $\semialgebraic(g_1, g_2)$ is empty, so $-1$ belongs to
  $\quadmod(g_1, g_2)$. We use the Algorithm $\algPutinar$ to compute its
  certificate as follows:

  \begin{enumerate}
  \item We construct the polynomial $g$ in Line 11 as
    $g = X^2 \cdot g_1 + (X + \frac{1}{2})^2 \cdot g_2 = -\frac{3}{2} -
    \frac{17}{4}X - 5 X^2 - \frac{1}{4}X^3 - X^4$.
\item The output in Line 12 of the algorithm is $(0, 0)$ as
    $\semialgebraic(g)$ is empty.
  \item We skip Lines 14, 15 as $\tilde{f} = -1$ is a global lower bound.
  \item The output in Line 16 is
    $\left(\frac{g - \gamma}{\gamma + \epsilon}\right)^{2N}$ where $\gamma =
    0$, $\epsilon= \frac{1}{4}$, and $N = 1$.
  \item The final certificate for $-1$ is $s_0 + s_1 \cdot g_1 + s_2 \cdot g_2$
    where 

    \begin{itemize}
    \item $s_0 =
      \frac{29117}{4}(-\frac{465}{29117}X^{6}-\frac{36296}{262053}X^{5}+\frac{50429}{174702}X^{4}+X^{3}+\frac{33127}{58234}X^{2}+\frac{12841}{524106}X-\frac{607}{29117})^2\\
      +\frac{16745322875}{4192848}(\frac{363342842}{16745322875}X^{6}-\frac{1495292264}{16745322875}X^{5}-\frac{10778946873}{16745322875}X^{4}+X^{2}+\frac{9257474543}{16745322875}X+\frac{5807084472}{83726614375})^2\\
      +\frac{73883223703533}{66981291500}(-\frac{11446181087671}{221649671110599}X^{6}+\frac{1229211466062151}{3989694079990782}X^{5}+X^{4}+\frac{1067029427027813}{3989694079990782}X+\frac{404381048939809}{6649490133317970})^2\\
      +\frac{1117740136583888157521}{4308869606390044560}(\frac{65437371804073029534}{1117740136583888157521}X^{6}+X^{5}\\
      -\frac{369072239463856436735}{1117740136583888157521}X-\frac{123666976097610416331}{1117740136583888157521})^2\\
      +\frac{1419243873302438492546149}{11177401365838881575210}(-\frac{33027302787456026609361}{2838487746604876985092298}X^{6}+X+\frac{1177306232301428228679629}{2838487746604876985092298})^2\\
      +\frac{4299616456470144258174085631}{510927794388877857316613640}(X^{6}+\frac{1464526336700556542015270577}{21498082282350721290870428155})^2\\
      +\frac{28772713500013993157761011337207}{19348274054115649161783385339500}$
    \item $s_1 = (X (6 + 17 X + 20 X^2 + X^3 + 4 X^4))^2$
    \item $s_2 = \frac{1}{4} ((1 + 2 X)(6 + 17 X + 20 X^2 + X^3 + 4 X^4))^2$
    \end{itemize}
  \end{enumerate}
\end{example1}

\begin{example1}
 Consider $G = \{g_1, g_2\}$ where $g_1 = -(X+1)^3$ and $g_2 = (X -
  1)^{3}$. We use the algorithm $\algPutinar$ to compute the certificate of $-1$
  in $\quadmod(G)$ as follows:

  \begin{enumerate}
  \item We construct the polynomial $g$ in Line 11 as $g = X^2 \cdot g_1 + (X -
    1)^{2} \cdot g_2 = -1 + 5 X - 11 X^2 + 7 X^3 - 8 X^4$
  \item The output in Line 12 of the algorithm is $(0, 0)$ as
    $\semialgebraic(g)$ is empty.
  \item We skip Lines 14, 15 as $\tilde{f} = -1$ is a global lower bound.
  \item The output in Line 16 is $\left(\frac{g - \gamma}{\gamma + \epsilon}\right)^{2N}$
    where $\gamma = 0, \epsilon = \frac{1}{4}$ and $N = 2$.
  \item The final certificate of $-1$ is $s_0 + s_1 \cdot g_1 + s_2 \cdot g_2$
    where
    \begin{itemize}
    \item $s_0 =
      \frac{17126597}{18}(-\frac{23414157}{342531940}X^{6}-\frac{24360705}{68506388}X^{5}+X^{4}-\frac{32228775}{68506388}X^{3}-\frac{36214623}{890583044}X^{2}+\frac{1551483}{34253194}X-\frac{77125}{17126597})^2\\+\frac{2845853977962779}{7124664352}(\frac{58824716915055}{406550568280397}X^{6}-\frac{18820203452463907}{25612685801665011}X^{5}+X^{3}-\frac{217008317257091}{406550568280397}X^{2}+\frac{1802315523100690}{25612685801665011}X+\frac{1428537132620}{2845853977962779})^2\\+\frac{566586418524550614601987}{4610283444299701980}(-\frac{441293949876084058754595}{1133172837049101229203974}X^{6}+X^{5}-\frac{4143062391045482797784955}{14731246881638315979651662}X^{2}+\frac{1798937864575798316360035}{14731246881638315979651662}X\\
-\frac{14420321768221338393045}{1133172837049101229203974})^2\\+\frac{4476823748424340803006975645797}{226634567409820245840794800}(X^{6}-\frac{185898767017264353856715669836315}{523788378565647873951816150558249}X^{2}+\frac{11440562998404149526041495015385}{58198708729516430439090683395361}X-\frac{1412446312023354092724091653625}{58198708729516430439090683395361})^2\\+\frac{460470784858936482413688371689482081326}{61283240292180801252362489615315133}(X^{2}\\-\frac{1795120297195923561216625168705161451155}{3683766278871491859309506973515856650608}X\\+\frac{615713974097244985604477960693254249353}{9209415697178729648273767433789641626520})^2\\+\frac{3253754066836303201506521383308561565538761447}{24902260045171284968932267140967190958110080}(X\\-\frac{996418825596998558913597005026911581215564210}{3253754066836303201506521383308561565538761447})^2\\+\frac{1102747962685314258229584020324746023825986273561231}{494895993565801716949141902401232214118445616088700}$
    \item $s_1 = 100 (X ((-1 + X)^5 - X^2 (1 + X)^3))^2$
    \item $s_2 = 100 ((-1 + X) ((-1 + X)^5 - X^2 (1 + X)^3))^2$
    \end{itemize}
  \end{enumerate}
\end{example1}

 \subsection{Examples}

\subsubsection{A simple example}

In this section, we discuss a simple example to highlight the key steps of our
approach. Consider the generators $G := \set{g_1, g_2}$, where
$g_1 := (X+3)(X+2)(X-1)$ and $g_2 := -(X+1)(X-2)(X-3)$. The semialgebraic set of
$G$ is $S = [-3, -2] \cup [2, 3]$. We will compute a certificate for the
polynomial $(X + 3)$, which is the left linear factor of $\Nat(S)$.

We compute a certificate for $X+3$ by splitting the generator $g_1$. The
semialgebraic set of $g_1$ is unbounded, thus we apply Algorithm
\ref{sec:addit-constr-from-4} to compute sums of squares $\sigma_1$ and $\tau_1$
and an auxiliary polynomial $h_1$ such that $1 = \sigma_1 g_1 + \tau_1 h_1$. We
obtain:

\begin{itemize}
\item $\sigma_1 := \frac{1}{2880}$
\item $\tau_1 := \frac{1}{2880}((X + \frac{17}{2})^2 + \frac{599}{4})$
\item $h_1 := -(X - 13)$
\end{itemize}

The semialgebraic set of $g_1 h_1$ is bounded. We can check 
$g_1 h_1 = \sigma_1 g_1^{2} \cdot h_1 + \tau_1 h_1^{2} \cdot g_1$. Since $h_1$
is strictly positive over $\semialgebraic(G)$, we find certificates in
$\quadmod(G)$ using the algorithm of Section \ref{sec:results--10};
$h_1 := s_0 + s_1 \cdot g_1 + s_2 \cdot g_2$ where $s_1 :=
\frac{589}{526513}X^2$, $s_2 := \frac{589}{526513}(X+5)^2$ and $s_0 := h_1 - s_1
\cdot g_1 - s_2 \cdot g_2$.

Therefore, the certificate of $g_1 h_1$ in $\quadmod(G)$ is
$$g_1 h_1 = \sigma_1 g_1^{2} s_0 +
 (\sigma_1 g_1^{2} s_1 + \tau_1 h_1^{2}) \cdot g_1
+ \sigma_1 g_1^{2} s_2 \cdot g_2$$ Next, using
$\algBasicLemma$ with inputs $X + 3, \frac{g_1}{X+3}h_1$ and $\{g_1 h_1\}$, we
find $\sigma_2$ and $\tau_2$ satisfying
$\sigma_2 (X+3) + \tau_2 \frac{g_1}{X+3} h_1$ :) :)

\begin{itemize}
\item $\sigma_2 := \frac{2543783}{78097856}X^2-\frac{16649805}{78097856}X
  +\frac{16870389}{39048928}-\frac{27573}{19524464}X^3-\frac{232}{1102409}\\
  (\frac{2543783}{78097856}X^2
  -\frac{16649805}{78097856}X+\frac{16870389}{39048928}
  -\frac{27573}{19524464}X^3)(X+3)(2-\frac{232}{1102409}(X+3)(-X^3+12X^2+15X-26)
  +(1-\frac{232}{1102409}(X+3)(-X^3+12X^2+15X-26))^2+(1-\frac{232}{1102409}(X+3)
  (-X^3+12X^2+15X-26))^3+(1-\frac{232}{1102409}(X+3)(-X^3+12X^2+15X-26))^4+(1
  -\frac{232}{1102409}(X+3)(-X^3+12X^2+15X-26))^5)(-X^3+12X^2+15X-26)$
\item
  $\tau_2 := \frac{889403}{78097856}-\frac{27573}{19524464}X
  +\frac{232}{1102409}(\frac{2543783}{78097856}X^2-\frac{16649805}{78097856}X
  +\frac{16870389}{39048928}-\frac{27573}{19524464}X^3)(X+3)^2(2
  -\frac{232}{1102409}(X+3)(-X^3+12X^2+15X-26)+(1
  -\frac{232}{1102409}(X+3)(-X^3+12X^2+15X-26))
  ^2+(1-\frac{232}{1102409}(X+3)(-X^3+12X^2+15X-26))^3+(1
  -\frac{232}{1102409}(X+3)(-X^3+12
  X^2+15X-26))^4+(1-\frac{232}{1102409}(X+3)(-X^3+12X^2+15X-26))^5)$
\end{itemize}

We can check $X + 3 = (X+3)^2 \sigma_2 + \tau_2 \cdot g_1 h_1$. Since
$\sigma_2$ and $\tau_2$ are strictly positive polynomials over
$\semialgebraic(g_1 h_1)$, we compute certificates of these in
$\quadmod(g_1 h_1)$. Using a square free decomposition, we notice that
$\sigma_2 := \sigma_{2_{sos}} \sigma_{2_{rest}}$ where
$\sigma_{2_{sos}} := (232X^4-2088X^3-11832X^2-4408X+1120505)^{6} \in
\sum{\qX}^2$ and $\sigma_{2_{rest}} :=
-\frac{3939}{5006538738653003018117160074353414225724432}X^3\\
+\frac{21023}{1158537889936232103365954397371038002646976}X^2\\
-\frac{16649805}{140183084682284084507280482081895598320284096}X\\
+\frac{16870389}{70091542341142042253640241040947799160142048}$ is strictly
positive over $\semialgebraic(g_1 h_1)$. We find certificates for the strictly
positive polynomial $\sigma_2$ using the algorithm $\algPutinar$ from
\cite{weifeng2025};
$\sigma_{2} := \sigma_{2_{sos}} s_3 + \sigma_{2_{sos}} s_4 \cdot g_1 h_1$ where:

\begin{itemize}
\item $s_3 := \sigma_{2_{rest}} - s_4 g_1 h_1$,
\item $s_4 := \frac{1}{1107706970296998652079462328093230026002541}\\
  +\frac{1}{75725680710000000000}(\frac{519162}{1690512571} g_1 h_1
  -\frac{1111230756}{1690512571})^{116}$
\end{itemize}

Using the algorithm $\algPutinar$, we also find
$\tau_2 = s_5 + s_6 \cdot g_1 h_1$ where

\begin{itemize}
\item $s_5 := \tau_2 - s_6 \cdot g_1 h_1$, 
\item $s_6 := \frac{1}{30574736731299198850326303450749} (\frac{2999}{1000} +
  X)^{22} \\
  +336370532867760591471442397905687904583764945\dots\\
  \dots 9211931020821776966282970246217703491 (\frac{57575}{240688632}g_1
h_1 -\frac{20539225}{40114772})^{304}$ 
\end{itemize}

Finally, the certificate of $X+3$ in $\quadmod(G)$ is:
\vspace{-4mm}

\begin{equation*}
  \label{sec:comp-cert-elem}
  \begin{split}
    X + 3
    &= (X+3)^2 \sigma_2 + \tau_2 \cdot g_1 h_1\\
    &=(X+3)^2 (\sigma_{2_{sos}} s_3 + \sigma_{2_{sos}} s_4 \cdot g_1 h_1) +
      (s_5 + s_6 \cdot g_1 h_1) g_1 h_1\\
    &=((X+3)^2\sigma_{2_{sos}} s_3 + s_6 (g_1 h_1)^{2}) +
      (\sigma_{2_{sos}} s_4 + s_5) \cdot g_1 h_1\\
    &=
      ((X+3)^2\sigma_{2_{sos}} s_3 + s_6 (g_1 h_1)^{2}) +
      (\sigma_{2_{sos}} s_4 + s_5)\sigma_1 g_1^{2} s_0\\
    &+
      (\sigma_{2_{sos}} s_4
      + s_5)(\sigma_1 g_1^{2} s_1 + \tau_1 h_1^{2}) \cdot g_1
      + (\sigma_{2_{sos}} s_4 + s_5)\sigma_1 g_1^{2} s_2\cdot g_2\\
  \end{split}
\end{equation*}

\subsubsection{Comparison with an alternative approaches}

In this section, we compare the certificates produced by our $\algSaturatedCert$
algorithm, the algorithm by Shang et al, and the Extended Gram matrix
approach. For the latter, we consider $G = \{g_1\}$ where
$g_1 := -(X + 2)(X + 1)(X - 1)(X - 3)$ and compute a certificate for different
input polynomials.

\begin{example1}
  \label{sec:comparison-with-an-1}
  Let the input polynomial be the left natural generator $X + 2$.

  \begin{itemize}
  \item $\algSaturatedCert$: We use $\algSOSBasicLemma$ with inputs $X + 2$ and
    $-(X + 1)(X - 1)(X - 3)$ to obtain sums of squares $\sigma, \tau$ such that
    $1 = \sigma (X + 2) + \tau (-(X + 1)(X - 1)(X - 3))$ where

    \begin{itemize}
    \item $\sigma = \frac{3}{5}(-\frac{5}{18}X+1)^2+\frac{11}{540}(X)^2$
    \item $\tau = \frac{1}{15}$
    \end{itemize}

    Hence, we obtain the certificate $X + 2 = s_0 + \frac{1}{15} \cdot g_1$
    where
    $s_0 = \frac{3}{5}((X +
    2)(-\frac{5}{18}X+1))^2+\frac{11}{540}((X+2)X)^2$. In expanded form, the
    degrees of the certificates are $\deg(s_0) = 4$ and $\deg(s_1) = 0$.
    
  \item \wgm: First, we fix the input polynomial $f$ at $x = -2$ with
    $f_1 := f - \frac{1}{15}\frac{(X - (-1))^{2}}{(-2 - (-1))^{2}}\frac{(X -
      1)^{2}}{(-2 - 1)^{2}}\frac{(X - 3)^{2}}{(-2 - 3)^{2}} \cdot g_1$ which is
    a sums of squares, i.e.,

    \begin{equation*} 
      \begin{split}
        f_1 &=
              \frac{63}{125}((X+2)(\frac{521}{559872}X^{4}
             +\frac{13513}{435456}X^{3}-\frac{61819}{483840}X^{2}-\frac{16}{63}X+1))^2\\ 
            &+\frac{2153239}{10080000}((X+2)(\frac{39646495}{2092948308}X^{4}-\frac{3092335}{38758302}X^{3}-\frac{558449}{2153239}X^{2}+X))^2\\
            &+\frac{21390631747541989}{428635813478400000}((X+2)\\
            &(\frac{21685151003276690}{577547057183633703}X^{4}-\frac{25691071038847010}{64171895242625967}X^{3}+X^{2}))^2\\
            &+\frac{202747529142388112939383}{89820118333198713490560000}((X+2)\\
            &(-\frac{923471009955226677313145}{3649455524562986032908894}X^{4}+X^{3}))^2\\
            &+\frac{59573723983908118887294869491}{13791788753274864784463186035584000}((X+2)(X^{4}))^2
      \end{split}
    \end{equation*}

    Hence, we obtain the certificates $X + 2 = s_2 + s_3 \cdot g_1$ where 

    \begin{itemize}
    \item $s_2 = f_1$
    \item $s_3 = \frac{1}{15}\frac{(X - (-1))^{2}}{(-2 - (-1))^{2}}\frac{(X -
        1)^{2}}{(-2 - 1)^{2}}\frac{(X - 3)^{2}}{(-2 - 3)^{2}}$
    \end{itemize}

    In expanded form, the degrees of the certificates are $\deg(s_2) = 10$ and
    $\deg(s_3) = 6$.

  \item Extended Gram matrix method: Let $m^{d} := \left(
      \begin{array}{c}
        1 \\
        X \\
        \vdots \\
        X^{d-1} \\
      \end{array}
    \right)$. There are no positive semidefinite matrices $A_0$ and $A_1$ of
    dimension $d \times d$ with $d < 3$ such that
    $f = ((m^{d})^{T} A_0 m^{d}) + ((m^{d})^{T} A_1 m^{d}) \cdot g_1$. However,
    there exist positive semidefinite matrices that satisfy these conditions with
    $d = 3$. We obtain the following.

    \begin{equation*} 
      \label{sec:comparison-with-an}
      \begin{split}
        A_0
        &= \left(
          \begin{array}{ccc}
            2.40002 & 0.533338 & -0.333341 \\
            0.533338 & 0.199996 & -0.0333376 \\
            -0.333341 & -0.0333376 & 0.0666673 \\
          \end{array}
          \right)\\
        A_1
        &= \left(
          \begin{array}{ccc}
            0.066669 & 4.526684 \times 10^{-7} & -0.000024 \\
            4.526684 \times 10^{-7} & 0.000050 & 4.720584 \times 10^{-8}\\
            -0.000024 & 4.720584 \times 10^{-8} & 2.817258 \times 10^{-8}\\
          \end{array}
          \right)
      \end{split}
    \end{equation*}

    The certificate obtained with this approach is 

    {
      \small
      \begin{equation*} 
        \begin{split}
          f &= ((m^{3})^{T} A_0 m^{3}) + ((m^{3})^{T} A_1 m^{3}) \cdot g_1\\
            &\approx (1.5492 +0.344267 X-0.21517 X^2)^2+(0. +0.000581784
              X^2)^2\\
            &\quad+(0. +0.285441 X+0.14272 X^2)^2\\
            &\quad+((0.258205 +1.75314*10^{-6} X-0.0000967477 X^2)^2\\
            &\quad\quad+(0. +0.00708567
              X+6.68609*10^{-6} X^2)^2+(0. +0.000136995 X^2)^2) \cdot g_1\\
        \end{split}
      \end{equation*}
    }

    In expanded form, the degrees of the certificates are $\deg((m^{3})^{T} A_0
    m^{3}) = 4$ and $\deg((m^{3})^{T} A_1 m^{3}) = 4$.
  \end{itemize}
\end{example1}

\begin{example1}
  Let the input polynomial be a natural generator $(X + 1)(X - 1)$.

  \begin{itemize}
  \item $\algSaturatedCert$: We use $\algSOSBasicLemma$ with inputs
    $(X + 1)(X - 1)$ and $-(X + 2)(X - 3)$ to obtain sums of squares
    $\sigma, \tau$ such that
    $1 = \sigma (X + 1)(X - 1) + \tau (-(X + 2)(X - 3))$ where

    \begin{itemize}
    \item
      $\sigma =
      \frac{1224263}{3896156}(-\frac{1518593}{9794104}X^{2}
      -\frac{723815}{14691156}X+1)^2
      +\frac{14940634371083}{686868427156032}(
      -\frac{10649700814773}{29881268742166}X^{2}+X)^2
      +\frac{360146266268039063}{931376675979220111168}(X^{2})^2
      + \frac{1323}{170573}((X + 2)(X - 3))^{2}$
    \item
      $\tau =
      \frac{1043377}{4093752}(-\frac{280155}{2086754}X^{2}
      -\frac{138821}{2086754}X+1)^2
      +\frac{222337281587}{17085306722016}(
      -\frac{105150010263}{222337281587}X^{2}+X)^2
      +\frac{76390416521127}{303397897057114808}(X^{2})^2 
      + \frac{10426}{974039}((X + 1)(X - 1))^{2}$
    \end{itemize}

    Hence, we obtain the certificate $(X + 1)(X - 1) = s_0 + s_1 \cdot g_1$ where

    \begin{itemize}
    \item $s_0 = (X+1)^2(X-1)^2 \sigma$
    \item $s_1 = \tau$
    \end{itemize}
    
    In expanded form, the degrees of the certificates are $\deg(s_0) = 8$ and
    $\deg(s_1) = 4$.

  \item \wgm: Let
    $t_n(a, b) := \sum_{i=0}^{n}{\frac{(-1)^{i}}{i!}  \frac{(X -
        a)^{2i}}{(2b)^{i}}}$ First, we fix the input polynomial $f$ at $x = -1$
    with
    $$f_1 := f - \frac{1}{4}t_{1}(-1, 4)^{2}\frac{(X - (-2))^{2}}{(-1 -
      (-2))^{2}} 
    \frac{(X - 1)^{2}}{(-1 - 1)^{2}} \frac{(X - 5)^2}{(-1 - 5)^2} \cdot g_1$$ 
    which is a local sums of squares in $\mathbb{R}[[X - 1]]$.

    We fix $f_1$ at $x = 1$ with

    $f_2 := f_1 - \frac{(X - 2)^{2}}{6}
    \frac{(X - (-2))^{2}}{(1 - (-2))^{2}}
    \frac{(X - (-1))^{2}}{(1 - (-1))^{2}}
    \frac{(X - 5)^2}{(1 - 5)^2}
    \cdot g_1$, which is a sums of squares. Therefore, the certificate
    obtained using this method is $(X + 1)(X - 1) = s_2 + s_3 \cdot g_1$ where

    \begin{itemize}
    \item $s_2 =
      \frac{2267}{4608}
      ((X+1)(X-1)(-\frac{3445}{1305792}X^{5}+\frac{47369}{3917376}X^{4}
      +\frac{12215}{217632}X^{3}-\frac{359423}{1450880}X^{2}
      -\frac{5875}{27204}X+1))^2 
      +\frac{13078917289}{30085447680}((X+1)(X-1)
      (\frac{277035625}{52315669156}X^{5}
      +\frac{828545605}{470841022404}X^{4}
      -\frac{7343687830}{39236751867}X^{3}
      +\frac{3249138825}{104631338312}X^{2}+X))^2
      +\frac{9964899093437107939}{666512004476200550400}
      ((X+1)(X-1)
      (\frac{256724554600819050}{9964899093437107939}X^{5}
      -\frac{176869388065834730}{1423557013348158277}X^{4}
      -\frac{2239275551240532900}{9964899093437107939}X^{3}+X^{2}))^2
      +\frac{175161724132613869398983}{952160936147766697988063232}((X+1)(X-1)
      (\frac{8243659133685020005983}{875808620663069346994915}X^{5}
      -\frac{124440064307804420379113}{525485172397841608196949}X^{4}+X^{3}))^2\\
      +\frac{7501703071687520501626007353}{188290838042063601193243027537920}
      ((X+1)(X-1)
      (-\frac{1488811505704947470507385048}{7501703071687520501626007353}X^{5}
      +X^{4}))^2\\
      +\frac{2330206160538201965573873681}{746665511493659640584241064664678400}
      ((X+1)(X-1)(X^{5}))^2$ 
    \item $s_3 =
      \frac{(X - 2)^{2}}{6}\frac{(X - (-2))^{2}}{(1 - (-2))^{2}}
      \frac{(X - (-1))^{2}}{(1 - (-1))^{2}} \frac{(X - 5)^2}{(1 - 5)^2}\\
      + \frac{1}{4}(t_{1}(-1, 4))^{2}\frac{(X - (-2))^{2}}{(-1 - (-2))^{2}}
      \frac{(X -1)^{2}}{(-1 - 1)^{2}}\frac{(X - 5)^2}{(-1 - 5)^2}$
    \end{itemize}

    In expanded form, the degrees of the certificates are $\deg(s_2) = 14$ and
    $\deg(s_3) = 10$.

  \item Extended Gram matrix method: We use the vector of monomials $m^{d}$ as
    in Example~\ref{sec:comparison-with-an-1}. There are no positive
    semidefinite matrices $A_0$ and $A_1$ of dimension $d \times d$ with $d < 4$
    such that $f = ((m^{d})^{T} A_0 m^{d}) + ((m^{d})^{T} A_1 m^{d}) \cdot
    g_1$. However, there exist positive semidefinite matrices that
    satisfy these conditions with $d = 4$. We obtain the following.

    \begin{equation*} 
      \label{sec:comparison-with-an-3}
      \begin{split}
        A_0
        &= \left(
          \begin{array}{cccc}
            0.10149 & -0.0332116 & -0.10149 & 0.0332115 \\
            -0.0332116 & 0.0247483 & 0.0332116 & -0.0247482 \\
            -0.10149 & 0.0332116 & 0.10149 & -0.0332115 \\
            0.0332115 & -0.0247482 & -0.0332115 & 0.0247481 \\
          \end{array}
          \right)\\
        A_1
        &= \left(
          \begin{array}{cccc}
            0.183582 & -0.020833 & -0.000437 & 0.000017
            \\
            -0.020833 & 0.025625 & -0.000016 &
                                               -0.000021 \\
            -0.000437 & -0.000016 & 0.000044 &
                                               2.766806 \times 10^{-8}\\
            0.000017 & -0.000021 &
                                   2.766806 \times 10^{-8}&
                                                            5.527610 \times 10^{-8}\\
          \end{array}
          \right)
      \end{split}
    \end{equation*}

    The certificate obtained with this approach is 

    {
      \small
      \begin{equation*} 
        \begin{split}
          f &= ((m^{4})^{T} A_0 m^{4}) + ((m^{4})^{T} A_1 m^{4}) \cdot g_1\\
            &\approx(0. +0.117814 X+2.53858*10^{-7} X^2-0.117813 X^3)^2\\
            &\quad+(0. +0.000294231
              X^2-4.34777*10^{-8} X^3)^2\\
            &\quad+(0. +0.00029423 X^3)^2+(0.318576 -0.10425 X-0.318575 X^2+0.10425
              X^3)^2\\
            &\quad+ ((0. +0.152517 X-0.000436259 X^2-0.000130442
              X^3)^2\\
            &\quad\quad+(0. +0.0065479 X^2+1.81908*10^{-6} X^3)^2\\
            &\quad\quad+(0.428464 -0.0486242
              X-0.00102061 X^2+0.0000403189 X^3)^2\\
            &\quad\quad+(0. +0.000191395 X^3)^2) \cdot g_1 
        \end{split}
      \end{equation*}
    }

    In expanded form, the degrees of the certificates are $\deg((m^{4})^{T} A_0
    m^{4}) = 6$ and $\deg((m^{4})^{T} A_1 m^{4}) = 6$.

  \end{itemize}
\end{example1}

 \section{Experiments} \label{experiments}

The algorithms discussed in Section \ref{cert_prods} have been implemented in \maple
as a package named \csq \footnote{\ csq is freely available on the following
  GitHub repository: \url{https://github.com/typesAreSpaces/SatQMCert}}. It
provides a procedure, \lpq, which uses natural generators and an arbitrary
polynomial as generators. The method first checks the membership of the input
polynomial in the quadratic module by checking non-negativity and returns a
\maple's table encoding generators in the quadratic module as indices and the
sum of squares multipliers as entries.
The implementation has been experimented with on several examples. 

The implementation has been compared with the implementation with \rc
\cite{10.1145/3282678.3282681}, a \maple package that computes
certificates of strictly positive elements in quadratic modules using a
semidefinite programming approach.

All experiments were performed on an M1 MacBook Air with 8GB of RAM running
macOS.

To the best of our knowledge, this is the first implementation that computes
certificates for saturated quadratic modules. Previous implementations have
focused on the case whenever the input polynomial is strictly positive
polynomials using Positivstellens\"atz's results. In particular, the latter
excludes polynomials sharing common zeros with the original generators without,
and most of the time, it is needed to include a polynomial in the generators
that witnesses the Archimedean property. A consequence is that these
implementations generate certificates in terms of the original generators, and
it is not necessary to include an Archimedean polynomial in the original
generators.

\subsection{Computing certificates of natural generators}

This benchmark uses a monogenic generator of the form
$p = -\prod_{i=1}^{k}{(X + i)(X-i)}$. We compute a certificate for the left
natural generator $X + k$ and a certificate for the right natural generator
$-(X - k)$. Table \ref{tab:comp-cert-natur-1} reports the time in seconds
required to complete this benchmark.

\begin{table}[htpb]
  \centering
  \begin{tabular}{ccc}
    \makecell{$k$}
    & \makecell{Left natural\\generator (seconds)}
    & \makecell{Right natural\\generator (seconds)}\\ 
    \hline
    1 & 0.101 & 0.035\\
    2 & 0.189 & 0.157\\
    3 & 0.295 & 0.297\\
    4 & 0.371 & 0.352\\
    5 & 0.524 & 0.534\\
    6 & 0.705 & 0.666\\
    7 & 0.903 & 0.847\\
    8 & 1.220 & 1.141\\
    9 & 1.523 & 1.345\\
    10 & 1.848 & 1.683\\
  \end{tabular}
  \caption{Computing certificates of products of natural generators}
  \label{tab:comp-cert-natur-1}
\end{table}

Additionally, \csq and \rc are compared on strictly positive polynomials. Table
\ref{tab:comp-cert-natur-2} reports the timings of the two packages.

\begin{table}[htpb]
  \centering
  \begin{tabular}{cccc}
    \makecell{$k$}
    & \makecell{Average time\\\csq (seconds)}
    & \makecell{Average time\\\rc (seconds)} \\
    \hline
    3 & 0.451 & 0.028 \\
    4 & 2.497 & 0.023 \\
    5 & 12.162 & 0.022 \\
    6 & 58.307 & 0.0155 \\
  \end{tabular}
  \caption{Computing certificates of strictly positive polynomials}
  \label{tab:comp-cert-natur-2}
\end{table}

 \section{Conclusions} \label{conclusion_sect}

The paper introduces a method for computing certificates of the membership of a
polynomial in an Archimedean univariate saturated quadratic module specified by a
finite set of generators. The method is novel and is based on computing
certificates in terms of the natural generators of a univariate saturated
quadratic module and combining the certificates of its natural generators in
terms of the original generators.  Given that a univariate saturated quadratic
module has a unique set of natural generators determined by its semialgebraic
set, the first step works for different sets of generators of the univariate
saturated quadratic module defined by the same semialgebraic set; this is one
of the advantages of the proposed method.  The second step changes for different
sets of generators, since certificates of natural generators must be computed in
terms of a given set of generators.

The method has been implemented in Maple and compared with \rc in a collection
of examples.

\section{Acknowledgments}
  We appreciate Prof. Chenqi Mou's and Weifeng Shang's insightful discussions on
  certificate computation in quadratic modules. We also thank anonymous
 reviewers who provided valuable feedback on this work.

\bibliographystyle{elsarticle-harv} 
\bibliography{references}

\appendix

\section{Implementation of Basic Lemma in univariate case}
\label{basiclemma-appendix}

This appendix includes an algorithmic description (Algorithm
\ref{sec:impl-basic-lemma}) of the constructive proof of
\th\ref{sec:results-}. The implementation is done using the programming
language $\maple$. The algorithm \ref{sec:impl-basic-lemma-sos} describes a
procedure for effectively constructing certificates for \th\ref{sec:results--1}.

\begin{algorithm}[ht]
  \DontPrintSemicolon
  \caption{Algorithm to find $\epsilon$ in Basic Lemma\\ $ \epsilon_{curr} :=
    \algFindEps(f, g, s_1, t_1, S, \delta) $} 
  \label{sec:impl-basic-lemma-findeps}
  \kwInput{$f, g, s_1, t_1 \in \aX$, $S \subseteq \mathbb{R}$, $\delta \in \mathbb{Q}$}
  \kwOutput{$\epsilon_{curr} \in \mathbb{Q}$}
  \kwRequires{$S$ is a non-empty finite union of intervals with algebraic
    numbers as end points.}

  \tcc{Find positive rational $\epsilon$ so that $f > 0, 1 > \epsilon \delta f$,
    and $\epsilon t_1 + s_1 f > 0$ over $L_2 = \{x \in S \mid g(x) \leq
    \epsilon\}$}

  $\epsilon_{top} := \ceil{\max_{x \in S}(g)}$\;
  $\epsilon_{curr} := \frac{\epsilon_{top}}{2}$\;

  \While{$\texttt{true}$}{
    $cond_1 := isEmpty(S \cap \semialgebraic(\epsilon_{curr} - g) \cap
    \semialgebraic(-f))$\; 
    $cond_2 := isEmpty(S \cap \semialgebraic(\epsilon_{curr} - g) \cap
    \semialgebraic(\epsilon_{curr} \delta f - 1))$\;
    $cond_3 := isEmpty(S \cap \semialgebraic(\epsilon_{curr} - g) \cap
    \semialgebraic(-(\epsilon_{curr}t_1 + s_1 f)))$\;
    \uIf{$cond_1$ and $cond_2$ and $cond_3$}{
      \Return $\epsilon_{curr}$\; 
    }
    $\epsilon_{curr} = \frac{\epsilon_{curr}}{2}$\;
  }
\end{algorithm}

\newpage

\begin{algorithm}[ht]
  \DontPrintSemicolon
  \caption{Algorithm to find $k$ in Basic Lemma\\ $ k :=
    \algFindK(f, g, s_1, t_1, S, \delta) $}
  \label{sec:impl-basic-lemma-findk}
  \kwInput{$f, g, s_1, t_1 \in \aX$, $S \subseteq \mathbb{R}$, $\delta \in
    \mathbb{Q}$} 
  \kwOutput{$k \in \mathbb{N}$}
  \kwRequires{$S$ is a non-empty finite union of intervals with algebraic
    numbers as end points.} 

  $\epsilon := \algFindEps(f, g, s_1, t_1, S, \delta)$\; 
  $L_2 := S \cap \semialgebraic(\epsilon - g)$\;
  $L_3 := S \cap \semialgebraic(g - \epsilon)$\;

  \tcc{Find $k$ such that $\epsilon t_1 + s_1 f > s_1 f (1 - \epsilon \delta
    f)^{k}$ over $L_2$ and $s_1 f (1 - \epsilon \delta f)^{k} < 1$ over $L_3$}
  $k := 0$\;

  \uIf{$L_2$ is non-empty}{
    \While{true}{
      $x^{*} := \argmax_{x \in L_2}(s_1 f (1 - \epsilon \delta f)^k)$\;  
      $value := (s_1 f (1 - \epsilon \delta f)^{k}) (x^{*})$\;
      $test := (\epsilon t_1 + s_1 f)(x^{*})$\;
      \uIf{$value < test$}{
        \kwBreak\;
      }
      \Else{
        $k := \left \lceil \frac{\log(\frac{test}{(s_1 f)(x^{*})})}{\log((1 -
            \epsilon  
            \delta f)(x^{*}))} \right \rceil$\;  
      }
    }
  }
  
  \uIf{$L_3$ is non-empty}{
    \While{true}{
      $x^{*} := \argmax_{x \in L_3}(s_1 f (1 - \epsilon \delta f)^{k})$\;
      $value := (s_1 f (1 - \epsilon \delta f)^{k})(x^{*})$\;
      \uIf{$value < 1$}{
        \kwBreak\;
      }
      \Else{
        $k := \left \lceil \frac{\log(\frac{1}{(s_1 f)(x^{*})})}{\log((1 -
            \epsilon \delta f)(x^{*}))} \right \rceil$\;
      }
    }
  }
  \Return $k$\;
\end{algorithm}

\begin{algorithm}[ht]
  \DontPrintSemicolon
  \caption{Implementation of Basic Lemma\\ $ \sigma, \tau := \algBasicLemma(f,
    g, G)$}
  \label{sec:impl-basic-lemma}
  \kwInput{$f, g \in \aX$, $G \subseteq \aX$}
  \kwOutput{$\sigma, \tau \in \aX$}
  \kwRequires{$\semialgebraic(G)$ is a union of bounded intervals
    $\bigcup_{i=0}^{k}{\left[a_i, b_i\right]}$; $f, g$ are non-negative over
    $\semialgebraic(G)$; $f, g$ are relatively prime.}
  \kwEnsures{$\sigma, \tau$ are strictly positive over $\semialgebraic(G)$ and
    $1 = \sigma f + \tau g$}
  
  $S := \semialgebraic(G)$\;

  Find $s, t \in \aX$ such that $1 = s f + t g$ using the extended Euclidean
  algorithm\;

  \uIf{$S$ is empty}{ 
    \Return $s, t$\;
  }

  $L_1 := S \cap \semialgebraic(-s)$\;

  \tcc{Find $N$ such that $s + N g > 0$ over $L_1$}

  \uIf{$isEmpty(L_1)$}{
    $N := 0$\;
  }
  \Else{
    \tcp{The polynomial $g$ is positive over $L_1$}
    $N := (1+\frac{1}{100})\left \lceil \max_{x \in L_1}(-\frac{s}{g}) \right
    \rceil$\;  
  }

  $s_1 := s + N g$\;
  $t_1 := t - N f$\;

  \tcc{Find positive rational $\delta$ such that $\delta f g < 1$ over $S$}

  $\delta := \frac{1}{\ceil{\max_{x \in S}(f g)}}$\;
  
  $k := \algFindK(f, g, s_1, t_1, S, \delta)$\;

  $r := s_1 \delta f \sum_{i=0}^{k-1}{(1 - \delta f g)^i}$\;
  $\sigma := s_1 - r g$\;
  $\tau := t_1 + r f$\;
  \Return $\sigma, \tau$\;
\end{algorithm}

\begin{algorithm}[ht]
  \DontPrintSemicolon
  \caption{Implementation of SOS variant of Basic Lemma\\ $ \sigma, \tau :=
    \algSOSBasicLemma(f, g)$}
  \label{sec:impl-basic-lemma-sos}
  \kwInput{$f, g \in \aX$}
  \kwOutput{$\sigma, \tau \in \aX$}
  \kwRequires{$f, g$ are relatively prime; $f, g$ are non-negative over
    $\semialgebraic(f g)$; $\semialgebraic(f g)$ is bounded}
  \kwEnsures{$\sigma, \tau$ are sums of squares and
    $1 = \sigma f + \tau g$}

  $\sigma_1, \tau_1 := \algBasicLemma(f, g, \{f g\})$\; 
  \uIf{$\sigma_1$ is a sums of squares}{
    $\sigma_{1,0}, \sigma_{1, 1} := \sigma_1, 0$\;
  }
  \Else{
    $\sigma_{1,0}, \sigma_{1, 1} := \algPutinar(\{f g\}, \sigma_1)$\;
  }
  \uIf{$\tau_1$ is a sums of squares}{
    $\tau_{1,0}, \tau_{1, 1} := \tau_1, 0$\;
  }
  \Else{
    $\tau_{1,0}, \tau_{1, 1} := \algPutinar(\{f g\}, \tau_1)$\;
  }
  $\sigma := \sigma_{1, 0} + \tau_{1, 1} g^{2}$\;
  $\tau := \tau_{1, 0} + \sigma_{1, 1} f^{2}$\;
  \Return $\sigma, \tau$\;
\end{algorithm}

 \section{Additional constructions from Section \ref{sec:natur-gener-assoc-1999}} 

In this section, we provide an algorithmic description of products in the
preordering structure. It is important to specify this operation since some
products might introduce square terms that are not included in the generators
associated with the preordering structure.

\begin{proposition}
  \th\label{sec:addit-constr-from} Let $G$ be the set of natural generators of
  some set of generators, $p_i = \sigma_i g_i$, and $p_j = \sigma_j g_j$ with
  $\sigma_i, \sigma_j \in \sum{\aX}^2$ and $g_i, g_j$ generators of
  $\PO(G)$. The certificates of the product $p_i p_j$ in $\PO(G)$ are
  computable.
\end{proposition}

\begin{proof}
  Let $\prod_{i=1}^{k}{a_i^{i}}$ be a square free factorization of $p_i
  p_j$. Let $s := \prod_{\substack{i=1 \\ i \in 2\mathbb{N}}}^{k}{a_i^{i}}$ and
  $t = \frac{p_i p_j}{s}$. Since $p_i$ and $p_j$ are generators of $\PO(G)$,
  we have that $s$ is of the form $\prod_{i}^{}{(X - a_i)^{2 n_i}}$. The
 generators $p$ in $\PO(G)$ only contain square terms when these do not modify
  the semialgebraic set, i.e., if $p = (X - a)^{2} q$, then
  $\semialgebraic(p) \neq \semialgebraic(q)$. Thus, we take square
  terms from $s$ that have this property and multiply them with $t$ in order
  construct a generator in $\PO(G)$. Let $\tilde{s}$ be the square terms in
  $s$ satisfying this property and let
  $s_1 := \frac{s}{\tilde{s}}, t_1 := \tilde{s} t$. Then a certificate of
  $p_i p_j$ in $\PO(G)$ is $(\sigma_i \sigma_j s_1) t_1$ where
  $\sigma_i \sigma_j s_1 \in \sum{\aX}^2$ and $t_1$ is a generator in
  $\PO(G)$.
\end{proof}

\begin{corollary}
  \th\label{sec:addit-constr-from-corr}
  Let $G$ be the set of natural generators of some set of generators.
  Let $p = \sum_{I \subseteq \set{1, \dots, s}}^{}{\sigma_I g_I}$, and $q =
  \sum_{I \subseteq \set{1, \dots, s}}^{}{\sigma^{'}_I g_I} \in \PO(G)$ with
  $\sigma_I, \sigma^{'}_I \in \sum{\aX}^2$. The certificates of $p q
  \in \PO(G)$ are computable.
\end{corollary}

\begin{proof}
  We distribute the terms in the product of $p_i p_j$ and use
  \th\ref{sec:addit-constr-from} to compute the certificates.
\end{proof}

 \section{Additional constructions from Section \ref{sec:cert-natur-gener-1982}}
\label{appendix-sec-4}

\begin{lemma} \th\label{sec:splitting-step-7} Let
  $p := \alpha(X-a)(X-b)^{m} + \beta(X-c)^{m}(X-d)$ with
  $m \in 2\mathbb{N} + 1$, $b, c \in \mathbb{A}$ within $(a, d)$, and
  $\alpha := \frac{1}{(d-a)(d-b)^{m}}, \beta := \frac{1}{(a-c)^{m}(a-d)}$. Then
  $p$ is of the form $q \cdot (X-a)(X-d) + 1$ for some $q \in \sum{\aX}^2$.
\end{lemma}

\begin{proof}
  Notice that $p(a) = p(d) = 1$. Using Euclid's division algorithm, the residue
  of $p$ divided by $(X-a)(X-d)$ is at most of degree 1, i.e.,
  $p = q(X-a)(X-d) + r$ with $\deg(r) \leq 1$. Let $r = \gamma_1 X +
  \gamma_2$. If $\gamma_1 = 0$ we are done as $\gamma_2 = 1$. Otherwise,
  $p(a) = r(a) = 1$, similarly $p(d) = r(d) = 1$. Hence, $r(a) = r(d)$, so
  $\gamma_1 a + \gamma_2 = \gamma_1 d + \gamma_2$, which implies that $a =
  d$. This contradicts the assumption that $a < d$.

  Now, we will show $q \in \sum{\aX}^2$ to conclude the proof. We will
  show that the only real roots of the polynomial $p - 1$ are $a, d$
  with multiplicity 1, which shows that $q$ must be a product of complex
  conjugate roots, which means that $q$ is a sum of squares.

  Suppose $q$ has a real root $e$ different from $a$ and $d$. This implies
 $p(e) = 1$. We will prove this leads to a contradiction by showing that $p$
  evaluates to a value greater than $1$ over $(-\infty, a)$ and $(d, \infty)$
  and smaller than $1$ over $(a, d)$.

  First, write $p$ as $p = \frac{p_1}{p_1(d)} + \frac{p_2}{p_2(a)}$ where
  $p_1 = (X-a)(X-b)^{m}$ and $p_2 = (X-c)^m(X-d)$. For the interval $x \in (d,
  \infty)$, we see that $\frac{p_1(x)}{p_1(d)}$ is strictly greater than $1$ and
  $\frac{p_2(x)}{p_2(a)}$ is non-negative; hence $p(x)$ is strictly greater than
  $1$. A similar argument holds for the interval $(-\infty, a)$.

  Now, let us consider the case where $x \in (a, d)$. If $b = c$ we have
 $p = \frac{p_2(a) + p_1(d)}{p_2(a) p_1(d)}(X-b)^{m}\left(X - r\right)$ where
  $r = \frac{a p_2(a) + d p_1(d)}{p_2(a) + p_1(d)}$. We can see that $p$ has at
  most two real zeros which belong to the interval $(a, d)$ as $a < r <
  d$. Since $p_1(d), p_2(a)$ are positive, the leading coefficient of $p$ is
  positive. The minimum value $p^{*}$ of $p$ is attained within
 $\left(a, d\right)$; this implies that $p$ strictly decreases in the
 interval $(a, p^{*})$ and increases strictly in the interval $(p^{*},
  d)$. As $p(a) = p(d) = 1$, we conclude that $p$ reaches a value less than
  $1$ over $(a, d)$.

  Now, let us consider the case where $b < c$. We can check that
  $p(b) = \frac{b-d}{a-d} (\frac{b-c}{a-c})^{m} < 1$ and
  $p(d) = \frac{a-c}{a-d} (\frac{b-c}{b-d})^{m} < 1$. To finish this case, we
  analyze the intervals $(a, b), (b, c)$ and $(c, d)$. Let us consider the
  interval $(a, b)$. The polynomial $\frac{p_1}{p_1(d)}$ evaluates to negative
  values over $(a, b)$ and the polynomial $\frac{p_2}{p_2(a)}$ is strictly
  decreasing over $(a, b)$ with maximum value $1$; thus $p$ evaluates to a value
  less than $1$ over the interval $(a, b)$; a similar reasoning handles the
  interval $(c, d)$.

  Now, we address the interval $(b, c)$. First, we notice that
  \vspace{-5mm}

  \begin{equation}
    \label{sec:addit-constr-from-1}
    \begin{split}
      \frac{dp}{dX}
      &= \frac{1}{p_1(d)}((X-b)^{m} + m (X-a)(X-b)^{m-1})\\
      &+
        \frac{1}{p_2(a)}((X-c)^{m} + m (X-d)(X-c)^{m-1})\\
      &= \frac{(X - b)^{m-1}}{(d-a)(d-b)^{m}}((m+1)X - (b + m a))\\
      &+
        \frac{(X-c)^{m-1}}{(a-d)(a-c)^{m}}((m+1)X - (c + m d))\\
    \end{split}
  \end{equation}

  We want to show $\frac{dp}{dX}$ is strictly increasing over $(b, c)$.  The
  polynomial $q_1 := \frac{(X - b)^{m-1}}{(d-a)(d-b)^{m}}((m+1)X - (b + m a))$
  has at most two real roots, $b$ and $\frac{b + ma}{m+1}$. The last is
 contained in the interval $[a, b]$, which means that the polynomial $q_1$ increases
 strictly for any point after $b$. Similarly, we can show that the
  polynomial $\frac{(X-c)^{m-1}}{(a-d)(a-c)^{m}}((m+1)X - (c + m d))$ is
  strictly increasing for any point before $c$. Hence, from equation
  \eqref{sec:addit-constr-from-1}, we see that $\frac{dp}{dX}$ is strictly
  increasing over $(b, c)$.

  We have
  $\frac{dp}{dX}(b) = \frac{1}{b-c} \frac{m(b-d) + (b - c)}{a-d}
  (\frac{b-c}{a-c})^{m} < 0$ as $b < c$. Similarly,
  $\frac{dp}{dX}(c) = \frac{1}{c-b} \frac{m(a-c) + (b - c)}{a-d}
  (\frac{c-b}{d-b})^{m} > 0$. By the Intermediate Value theorem, we have that
  $\frac{dp}{dX}$ has a zero $\eta$ in $(b, c)$ and it is unique since
  $\frac{dp}{dX}$ is strictly increasing over $(b, c)$. Furthermore, $\eta$
  corresponds to the global minimum value of $p$ over $\mathbb{R}$. Hence
  $p(x) < 1$ for $x \in (b, c)$, as otherwise there would be a local maximum
  value of $p$ in $(b, c)$, which is not the case.

  Finally, let us consider the case where $c < b$. For the interval $(a, c)$,
  $\frac{p_1}{p_1(d)}$ is less than $1$ and $\frac{p_2}{p_2(a)}$
  is always non-negative over $(a, c)$, hence $p$ is less than $1$ over $(a,
  c)$. A similar argument handles the interval $(b, d)$. For the interval $[c,
  b]$ we notice that both $\frac{p_1}{p_1(d)}$ and $\frac{p_2}{p_2(a)}$ are
  non-positive, so $p$ is less than $1$ over $[b, d]$. Therefore, $p$ is less
  than $1$ over $(a, d)$.

  Therefore, $q$ has no real roots and is a product of irreducible quadratics
  of the form $q = \prod_{j}^{}{((X - a_j)^2 + b_j^2)^{e_j}}$. Since the leading
  coefficient of $p$ is positive, we see that $q$ is a sum of squares.
\end{proof}

\begin{lemmar}{\ref{sec:splitting-step-5}}
  Let $a_0, b_i, c_{i_1}, c_{i_2}, a_{i+1}, b_k \in \mathbb{A}$ such that
  $a_0 \leq b_i < c_{i_2} < c_{i_1} < a_{i+1} \leq b_k$ and $m_{i_1}, m_{i_2}$
  odd numbers. We have
  $(X - b_i)(X - a_{i+1}) \in Q := \quadmod(X-a_0, (X - b_i)(X -
  c_{i_1})^{m_{i_1}}, (X - c_{i_2})^{m_{i_2}}(X- a_{i+1}), -(X - b_k))$ and its
  certificates in $Q$ are computable.
\end{lemmar}

\begin{proof}
  Let $g_{\mathrm{nat}}$ denote the natural generator
  $(X-b_i)(X-a_{i+1}) \in \Nat(S)$ and let
  $G_{\mathrm{fix}} := \set{X-a_0, (X - b_i)(X - c_{i_1})^{m_{i_1}}, (X -
    c_{i_2})^{m_{i_2}}(X - a_{i+1}), -(X - b_k)}$. The following is a high-level
  structure of this proof.

  (i) Find a polynomial $q \in \quadmod(G_{\mathrm{fix}})$ of the form
  $q = g_{\mathrm{nat}}(\alpha \cdot g_{\mathrm{nat}} + 1)$ with
  $\alpha \in \sum{\aX}^2$.

  (ii) As $g_{\mathrm{nat}}$ is non-negative over
  $\semialgebraic(G_{\mathrm{nat}})$, we use the algorithm $\algBasicLemma$ with
  inputs $g_{\mathrm{nat}}, \alpha \cdot g_{\mathrm{nat}} + 1$ and
  $G_{\mathrm{fix}}$, compute strictly positive polynomials $\sigma, \tau$ over
  $\semialgebraic(G_{\mathrm{fix}})$ such that
  $1 = \sigma g_{\mathrm{nat}} + \tau (\alpha \cdot g_{\mathrm{nat}} + 1)$.

  (iii) Compute certificates for $\sigma, \tau$, and $\sigma \tau$ in
  $\quadmod(G_{\mathrm{fix}})$ using the algorithm $\algPutinar$.

  (iv) Using the previous two steps, compute a certificate for
  $\tau g_{\mathrm{nat}}$ and $g_{\mathrm{nat}}$ in
  $\quadmod(G_{\mathrm{fix}})$.

  Consider
  \begin{equation}
    \label{sec:splitting-step-8}
    \begin{split}
      q_1 &:= \alpha_1 g_{\mathrm{nat}}^2 \cdot (X - a_0) \\
          &+ \alpha_2(X-a_{i+1})^2 \cdot (X - b_i)(X-c_{i_1})^{m_{i_1}}\\
          &+ \alpha_3(X-b_i)^2 \cdot (X - a_{i+1})(X-c_{i_2})^{m_{i_2}} \\
          &+ \alpha_1 g_{\mathrm{nat}}^2 \cdot (-(X - b_k))\\
    \end{split}
  \end{equation}

  {\noindent with $\alpha_1, \alpha_2, \alpha_3 \in \mathbb{A}^{+}$, $q$ in
    step 1 above is constructed as follows.}
  We can assume $m_{j_1} = m_{j_2}$ as otherwise, $\alpha_2$ and $\alpha_3$ are
  replaced by
  $\alpha_2(X-c_{i_1})^{\max(m_{j_1}, m_{j_2}) - \min(m_{j_1}, m_{j_2})}$ and
  $\alpha_3(X-c_{i_2})^{\max(m_{j_1}, m_{j_2}) - \min(m_{j_1}, m_{j_2})}$,
  respectively, since $m_{j_1}, m_{j_2}$ are odd numbers, $q_1$ can be rewritten
  as

  \begin{equation*}
    \label{sec:merging-step}
    \begin{split}
      q_1 &=
      g_{\mathrm{nat}}(\alpha_1 \cdot g_{\mathrm{nat}} ((X-a_0) +
        (-(X-b_k)))\\
      &+\alpha_2 \cdot (X-a_{i+1})(X-c_{i_1})^{m_{i_2}}
        +\alpha_3 \cdot (X-b_i)(X-c_{i_2})^{m_{i_2}})\\
      &= g_{\mathrm{nat}}(\alpha_1(b_k - a_0) \cdot g_{\mathrm{nat}}\\
      &+\alpha_2 \cdot (X-a_{i+1})(X-c_{i_1})^{m_{i_2}}
        +\alpha_3 \cdot (X-b_i)(X-c_{i_2})^{m_{i_2}})\\
    \end{split}
  \end{equation*}

  Make $\alpha_2$ to be $\frac{1}{(b_i-a_{i+1})(b_i-c_{i_1})^{m_{i_2}}}$ and
  $\alpha_3$ to be $\frac{1}{(a_{i+1}-b_i)(a_{i+1}-c_{i_2})^{m_{i_2}}}$, which
  are positive constants, since $b_i < c_{i_1} \leq c_{i_2} < a_{i+1}$. By
  \th\ref{sec:splitting-step-7} from below,
  $\alpha_2(X-b)(X-c_{j_2})^{m_{j_1}} +\alpha_3(X-a)(X-c_{j_1})^{m_{j_1}}$ is of
  the form $q_2 \cdot (X-a)(X-b) + 1$ for some $q_2 \in \sum{\aX}^2$. Hence,
  $q_1$ is simplified to
  \vspace{-4mm}

  \begin{equation}
    \label{sec:splitting-step-11}
    \begin{split}
      q_1
      &=g_{\mathrm{nat}}(\alpha_1(b_k - a_0) \cdot g_{\mathrm{nat}}
      + q_2 \cdot g_{\mathrm{nat}} + 1)\\
      &=g_{\mathrm{nat}}((\alpha_1(b_k - a_0) + q_2) \cdot g_{\mathrm{nat}} + 1)
    \end{split}
  \end{equation}

  \noindent
  The resulting polynomial $q_1$ is therefore the required polynomial $q$ above.

  Let $q_3 := \alpha_1(b_k - a_0) + q_2 \in \sum{\aX}^2$. Since
  $g_{\mathrm{nat}}$ and $q_3 \cdot g_{\mathrm{nat}} + 1$ are relatively prime,
  we use $\algBasicLemma$ with inputs
  $g_{\mathrm{nat}}, q_3 \cdot g_{\mathrm{nat}} + 1$ and $G_{\mathrm{fix}}$ to
  find $\sigma, \tau$ such that both are strictly positive over
  $\semialgebraic(G_{\mathrm{fix}})$ and
  \begin{equation}
    \label{sec:merging-step-4}
    \begin{split}
      1 = \sigma g_{\mathrm{nat}} + \tau (q_3 \cdot g_{\mathrm{nat}} + 1)
    \end{split}
  \end{equation}

 Additionally, $\sigma \tau$ is also strictly positive over
  $\semialgebraic(G_{\mathrm{fix}})$. The ertificates for $\sigma, \tau$ and
 $\sigma \tau$ in $\quadmod(G_{\mathrm{fix}})$ are computed using the
  algorithm $\algPutinar$ from \cite{weifeng2025}.

  We multiply $\tau g_{\mathrm{nat}}$ on both sides of equation
  \eqref{sec:merging-step-4} to obtain
  $\tau g_{\mathrm{nat}} = g_{\mathrm{nat}}^2 \cdot \sigma \tau + \tau^2 \cdot
  q$, for which the certificates are given explicitly. Multiplying both sides of
  equation \eqref{sec:merging-step-4} by $g_{\mathrm{nat}}$ we obtain
  $g_{\mathrm{nat}} = g_{\mathrm{nat}}^2 \sigma + g_{\mathrm{nat}}^{2} q_3 \cdot
  \tau + \tau g_{\mathrm{nat}}$, thus the certificates for $g_{\mathrm{nat}}$
  are given explicitly in $\quadmod(G_{\mathrm{fix}})$.
\end{proof}

\begin{lemmar}{\ref{sec:cert-quadr-fact}}
  Let $a_0, b_i, c_{i_1}, c_{i_2}, a_{i+1}, b_k \in \mathbb{A}$ such that
  $a_0 \leq b_i < c_{i_1} \leq c_{i_2} < a_{i+1} \leq b_k$ and
  $m_{i_1}, m_{i_2}$ odd numbers. We have
  $(X - b_i)(X - a_{i+1}) \in Q := \quadmod(X-a_0, (X - b_i)(X -
  c_{i_1})^{m_{i_1}}, (X - c_{i_2})^{m_{i_2}}(X- a_{i+1}), -(X - b_k))$ and its
  certificates in $Q$ are computable.
\end{lemmar}

\begin{proof}
  Let
  $G_{\mathrm{fix}} := \{X-a_0, (X - b_i)(X - c_{i_1})^{m_{i_1}}, (X -
  c_{i_2})^{m_{i_2}}(X- a_{i+1}), -(X - b_k)\}$.  We use the results in Section
  \ref{cert_prods} to find certificates of the product of
  $(X - b_i)(X - c_{i_1})$ and $(X - c_{i_2})(X- a_{i+1})$ in
  $\quadmod(X-a_0, (X - b_i)(X - c_{i_1}), (X - c_{i_2})(X- a_{i+1}), -(X -
  b_k))$, i.e.,
  $(X - b_i)(X - c_{i_1})(X - c_{i_2})(X- a_{i+1}) = \sigma_0 + \sigma_1 \cdot
  (X - a_0) + \sigma_2 \cdot (X - b_i)(X - c_{i_1}) + \sigma_3 \cdot (X -
  c_{i_2})(X- a_{i+1}) + \sigma_4 \cdot (-(X - b_k))$. Since $m_1, m_2$ are odd
  numbers, we multiply the sums of squares
  $(X - c_{i_1})^{m_{i_1} - 1} (X - c_{i_2})^{m_{i_2} - 1}$ to the last
  equation; thus we obtain a certificate of
  $p := (X - b_i)(X - c_{i_1})^{m_{i_1}}(X - c_{i_2})^{m_{i_2}}(X- a_{i+1})$ in
  $\quadmod(G_{\mathrm{fix}})$.

  Since we have certificates for the last expression, we use the algorithm
  $\algRemoveStrictPosBet$ to remove the interval $[c_{i_1}, c_{i_2}]$ from the
  semialgebraic set of $p$. Finally, we use the result in Section
  \ref{sec:cert-non-neg-cert} to obtain a certificates of $(X - b_i)(X - a_{i+1})$
  in $\quadmod(G_{\mathrm{fix}})$.
\end{proof}

\begin{proposition}
  \th\label{sec:addit-constr-from-3}

  Let $f$ be a polynomial such that
  $\semialgebraic(f) = \bigcup_{i=0}^{k}{[c_i, d_i]}$ is bounded, divisible by
 $(X - d_j)^{m}(X - c_{j+1})^{n}$, where $m$ is odd and $n$ is an even
  number. Then $(X - d_j)^{m}(X - c_{j+1})^{n + 1}$ belongs to
  $\quadmod(f)$. Furthermore, the certificates of
  $(X - d_j)^{m}(X - c_{j+1})^{n+1}$ in $\quadmod(f)$ are computable.
\end{proposition}

\begin{proof} First, we notice that $\epsilon_{d_j}(f) = -1$ as $d_j$ is a right
  end point. Let
  $I = \{i \in \mathbb{N} \mid j + 1 < i \leq k, \ord_{c_i}(f) \text{ is odd }
  \}$ be a set of indices.

  If $I$ is empty, it means that $f$ only has end points of even multiplicity to
  the right of $c_{j+1}$. If so, we define the polynomial
  $s := (X - d_j)^{m-1}\prod_{i=j+1}^{k}{(X - c_i)^{\ord_{c_i}(f)}}$ and define
  $\hat{f} := \frac{f}{s}$. Notice that $s$ is a sums of squares so the
  semialgebraic set of $\hat{f}$ is $\semialgebraic(f)$ without the isolated
  points $c_{j+1}, \dots, c_k$. Additionally, the multiplicity of $(X - d_j)$ in
  $\hat{f}$ is 1. The polynomial $-(X - d_j)$ is non-negative over
  $\semialgebraic(\hat{f})$, therefore we compute its certificate using the
  procedure in Section \ref{sec:cert-non-neg-cert} in $\quadmod(\hat{f})$. The
  polynomial $-(X - c_{j+1})$ is strictly positive over $\quadmod(\hat{f})$, so
  we compute its certificate using the algorithm $\algPutinar$. Since
  $\quadmod(\hat{f})$ is a preorder, we can obtain a certificate of
  $(X - d_j)(X - c_{j+1})$ in $\quadmod(\hat{f})$ by rearranging the
  certificates of $-(X - d_j)$ and $-(X - c_{j+1})$. Let
  $(X - d_j)(X - c_{j+1}) = \sigma_0 + \sigma_1 \cdot \hat{f}$. Multiplying $s$
 by the previous equation, we obtain
  $(X - d_j)^{m}(X - c_{j+1})^{n+1} \prod_{i=j+2}^{k}{(X - c_i)}^{\ord_{c_i}(f)}
  = s \sigma_0 + \sigma_1 \cdot f$. Finally, we compute a certificate of
  $(X - d_j)^{m}(X - c_{j+1})^{n+1}$ in $\quadmod(f)$ using the procedure in
  Section \ref{sec:cert-non-neg-cert}.

  If $I$ is not empty, let $l \in I$ be the smallest one. In this case, we
  define
  $s := (X-d_j)^{m-1} \prod_{i=j+1}^{l-1}{(X - c_i)^{\ord_{c_i}(f)}}(X -
  c_l)^{\ord_{c_l}(f)-1}$. We compute a certificate of $(X - d_j)(X - c_l)$ in
  $\quadmod(\hat{f})$ using the procedure in Section
  \ref{sec:cert-non-neg-cert}. Then, using \th\ref{sec:cert-elem-ponats}, we
  compute the certificate of $(X - d_j)(X - c_{j+1})$ in $\quadmod((X - d_j)(X -
  c_l))$, lifting the certificate to $\quadmod(\hat{f})$ as we have a certificates
  of $(X - d_j)(X - c_l)$ in $\quadmod(\hat{f})$. We repeat the same procedure
  of multiplying $s$ from the case when $I$ is empty in order to compute
  certificates of $(X - d_j)^{m}(X - c_{j+1})^{n+1}$ in $\quadmod(f)$.
\end{proof}

\end{document}